\documentclass{amsart}

\usepackage[margin=1in]{geometry}

\usepackage{wrapfig}
\usepackage{comment}
\usepackage{xcolor}
\usepackage{mathbbol}
\usepackage{anyfontsize}  
\usepackage{thmtools}
\usepackage{thm-restate}
\usepackage{mathtools}
\usepackage{amsmath}
\usepackage{amsthm}
\usepackage{amssymb}
\usepackage{tikz}
\usepackage{tikz-cd}
\usetikzlibrary{matrix,arrows,cd,calc}
\usepackage{hyperref}
\hypersetup{
    hidelinks,
}
\usepackage{xcolor}
\definecolor{darkgreen}{rgb}{0,0.45,0}
\hypersetup{colorlinks,urlcolor=blue,citecolor=darkgreen,linkcolor=darkgreen,linktocpage}
\usepackage[capitalize]{cleveref}
\usepackage{xspace}
\usepackage{caption}
\usepackage{subcaption}
\usepackage{enumerate}
\usepackage{centernot}
\usepackage{textcomp}
\usepackage{thmtools}
\usepackage[noend,ruled,linesnumbered]{algorithm2e}
\usepackage{algpseudocode}
\usepackage{enumitem}

\newtheorem{theorem}{Theorem}[section]
\newtheorem{lemma}[theorem]{Lemma}
\newtheorem{proposition}[theorem]{Proposition}

\theoremstyle{definition}
\newtheorem{definition}[theorem]{Definition}
\newtheorem{example}[theorem]{Example}

\newtheorem{construction}[theorem]{Construction}

\newtheorem{notation}[theorem]{Notation}
\newtheorem*{notation*}{Notation}

\crefname{algocf}{Algorithm}{Algorithms}

\renewcommand{\paragraph}[1]{\medskip\noindent\textbf{#1.}}

\def\noteson{\gdef\luis##1{\noindent{\color{blue}[Luis: ##1]}}
\gdef\dmitriy##1{\noindent{\color{violet}[Dmitriy: ##1]}}
\gdef\todo##1{\noindent{\color{red}[todo: ##1]}}}

\noteson

\renewcommand{\to}{\xrightarrow{\;\;\;}}

\renewcommand{\epsilon}{\varepsilon}
\renewcommand{\phi}{\varphi}

\newcommand{\mergeattime}{\mathsf{merge\_at\_time}}
\newcommand{\timeofmerge}{\mathsf{time\_of\_merge}}

\newcommand{\append}{\mathsf{append}}

\let\ker\undefined
\DeclareMathOperator{\ker}{ker}
\DeclareMathOperator{\coker}{coker}
\DeclareMathOperator{\Hom}{Hom}

\DeclareMathOperator{\Ext}{Ext}

\newcommand{\Hil}{\mathsf{Hil}}

\newcommand{\vect}{\mathbf{Vec}}
\newcommand{\set}{\mathbf{Set}}
\newcommand{\Ab}{\mathrm{Ab}}

\newcommand{\op}[1]{{\bf #1}}
\newcommand{\var}[1]{{\sf #1}}

\DeclareMathAlphabet{\mathpzc}{OT1}{pzc}{m}{it}

\usepackage[mathscr]{euscript}

\makeatletter
\newcommand\DEFINEALPHABETLOOP[3]{%
  \ifx\relax#3\expandafter\@gobble\else\expandafter\@firstofone\fi
  {\expandafter\newcommand\expandafter*\csname#3#1\endcsname{#2{#3}}%
   \DEFINEALPHABETLOOP{#1}{#2}}%
}%
\newcommand\Definealphabet[2]{%
  \DEFINEALPHABETLOOP{#1}{#2}abcdefghijklmnopqrstuvwxyzABCDEFGHIJKLMNOPQRSTUVWXYZ\relax
}%
\makeatother
\Definealphabet{bb}{\mathbb}
\Definealphabet{cal}{\mathcal}
\Definealphabet{bf}{\mathbf}
\Definealphabet{sf}{\mathsf}
\Definealphabet{rm}{\mathrm}
\Definealphabet{frak}{\mathfrak}
\Definealphabet{scr}{\mathscr}
\title[Minimal Presentations of Zero-dimensional Persistent Homology]{Computing Betti Tables and Minimal Presentations of Zero-dimensional Persistent Homology}

\author{Yuan Luo}
\address{University of California, Davis, CA, USA;
Lawrence Berkeley National Laboratory, Berkeley, CA, USA}
\email{uluo@ucdavis.edu}

\author{Dmitriy Morozov}
\address{International Computer Science Institute, Berkeley, CA, USA;
Lawrence Berkeley National Laboratory, Berkeley, CA, USA}
\email{dmitriy@mrzv.org}

\author{Luis Scoccola}
\address{Centre de Recherches Mathématiques et Institut des sciences mathématiques;
Laboratoire de combinatoire et d'informatique mathématique de l'Université du Québec à Montréal;
Université de Sherbrooke; Québec, Canada}
\email{luis.scoccola@gmail.com}

\begin{document}

\maketitle


\begin{abstract}
The Betti tables of a multigraded module encode the grades at which there is an algebraic change in the module.
Multigraded modules show up in many areas of pure and applied mathematics, and in particular in topological data analysis, where they are known as persistence modules, and where their Betti tables describe the places at which the homology of filtered simplicial complexes changes.
Although Betti tables of singly and bigraded modules are already being used in applications of topological data analysis, their computation in the bigraded case (which relies on an algorithm that is cubic in the size of the filtered simplicial complex) is a bottleneck when working with large datasets.
We show that, in the special case of zero-dimensional homology (particularly relevant for clustering), Betti tables of bigraded modules can be computed in log-linear time.
We also consider the problem of computing minimal presentations, and show that minimal presentations of zero-dimensional persistent homology can be computed in quadratic time, regardless of the grading~poset.
Our algorithms are independent of the field of coefficients, which implies that the zeroth and first Betti tables of zero-dimensional persistence homology are independent of the coefficients; we show that, remarkably, this is not true for higher Betti tables.
We implement our algorithms and demonstrate that they are orders of magnitude faster than existing methods on the same datasets, and can scale to significantly larger inputs without exceeding memory limits.
\end{abstract}


\section{Introduction}

\paragraph{Betti tables and persistence}
Betti tables are a classical descriptor of a multigraded modules~\cite{eisenbud,miller-sturmfels,peeva}, which encode the grades of the generators in a minimal projective resolution of the module (see, e.g., \cref{fig:graded-graph}
and Example \ref{example:resolution-of-H_0}).
Informally, one can interpret the Betti tables of a graded module as recording the grades at which there is an algebraic change in the module.
Graded modules have applications in a wide variety of areas of pure and applied mathematics, including topological data analysis, and more specifically, persistence theory \cite{oudot,botnan-lesnick}, where they are known as \emph{persistence modules}, and where they are used to describe the varying topology of simplicial complexes and other spaces as they are filtered by one or more real parameters.

Informally, \emph{one-parameter persistence modules} correspond to $\Zbb$-graded $\kbb[x]$-modules, and can thus be classified up to isomorphism effectively.
\emph{Multiparameter persistence modules}~\cite{carlsson-zomorodian,botnan-lesnick} correspond to $\Zbb^n$-graded $\kbb[x_1, \dots, x_n]$-modules, and thus do not admit any reasonable classification up to isomorphism (formally, one is dealing with categories of wild representation type~\cite{Nazarova}; see~\cite{carlsson-zomorodian,bauer-botnan-oppermann-steen,bauer-scoccola} for manifestations of this phenomenon in persistence theory).
For this reason, much of the research in multiparameter persistence is devoted to the study of incomplete descriptors of multiparameter persistence modules.
Betti tables (also known as multigraded Betti numbers) provide one of the simplest such descriptors, and various properties of this descriptor from the point of view of persistence theory are well understood, including
their effective computation \cite{lesnick-wright,kerber-rolle,fugacci-kerber-rolle,bauer-lenzen-lesnick},
their relationship to discrete Morse theory~\cite{guidolin-landi,allili-kaczynski-landi,guidolin-landi-2},
their optimal transport Lipschitz-continuity with respect to perturbations~\cite{oudot-scoccola},
and their usage in supervised learning~\cite{loiseaux-et-al,scoccola-et-al}.

\paragraph{Two-parameter persistent homology}
The simplest case beyond the one-parameter case (i.e., the singly graded case) is the two-parameter case (i.e., the bigraded case).
Here, one is usually given a finite simplicial complex~$K$ together with a function $f : K \to \Rbb^2$ mapping the simplices of $K$ to~$\Rbb^2$, which is monotonic, i.e., such that $f(\sigma) \leq f(\tau)$ whenever $\sigma \subseteq \tau \in K$ (see \cref{fig:graded-graph} for an example).
By filtering $K$ using $f$ and taking homology in dimension $i \in \Nbb$ with coefficients in a field~$\kbb$, one obtains an~$\Rbb^2$-graded module, or equivalently, a functor $H_i(K,f; \kbb) : \Rbb^2 \to \vect_\kbb$.
Examples include geometric complexes of point clouds filtered by a function on data points, such as a density estimate~\cite{cai-kim-memoli-wang,carriere-blumberg}, and graphs representing, say, molecules or networks, filtered by two application-dependent quantities~\cite{todd}.
Applying homology to a bifiltered simplicial complex is justified by the fact that the output is automatically invariant under relabeling, meaning that any operation based on this module, such as computing its Betti tables, will result in a relabeling-invariant descriptor.
In this setup, one of the main stability results of~\cite{oudot-scoccola} implies that, for $f,g : K \to \Rbb^2$ any two monotonic functions, we have $\| \mu_{f}  - \mu_{g} \|_1^{\mathsf{K}} \;\leq \; 2 \cdot \|f - g\|_1$,
where $\|-\|_1^{\mathsf{K}}$ denotes the Kantorovich--Rubinstein norm between signed measures (also known as the $1$-Wasserstein distance), and $\mu_f$ and $\mu_g$ are the signed measures on $\Rbb^2$ obtained as the alternating sum of Betti tables of $H_i(K,f;\kbb)$ and $H_i(K,g;\kbb)$, respectively; see \cite[Theorem~1]{loiseaux-et-al} for details.
The upshot is that the Betti tables of the homology of bifiltered simplicial complexes form a perturbation-stable, relabeling-invariant descriptor of bifiltered simplicial complexes.

\begin{figure}
    \includegraphics[width=0.9\linewidth]{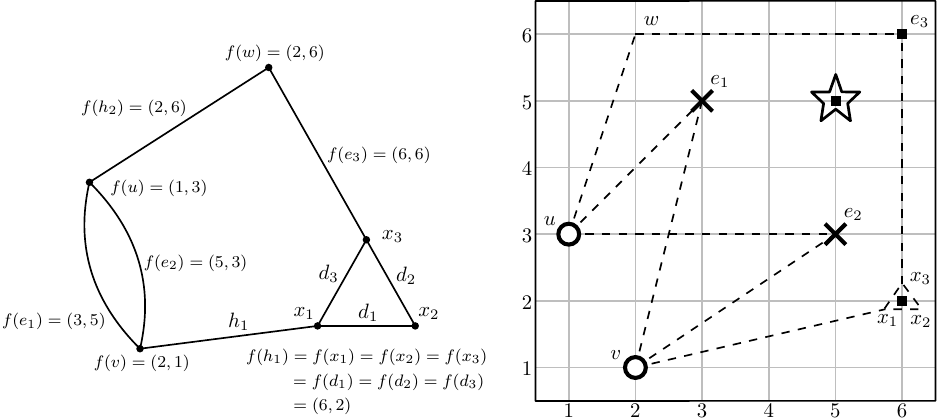}
    \caption{%
    \emph{Left.}
    A bifiltered graph $(G,f)$ with vertex set $\{u,v,w,x_1,x_2,x_3\}$ and edge set $\{e_1, e_2, e_3, h_1, h_2, d_1, d_2, d_3\}$.
    \emph{Right.}
    The bifiltered graph schematically mapped to $\Rbb^2$, together with the Betti tables $\beta_0(H_0(G,f))$~(circles), $\beta_1(H_0(G,f))$~(crosses), $\beta_2(H_0(G,f))$~(stars), and $\beta_0(H_1(G,f))$~(squares).}
    \label{fig:graded-graph}
\end{figure}

The current standard algorithm for computing the Betti tables of $H_i(K,f;\kbb)$ in the two-parameter case is the Lesnick--Wright algorithm~\cite{lesnick-wright}, which runs in time $O(|K|^3)$.
Because of results such as those of \cite{edelsbrunner-parsa}, one does not expect to find algorithms with better worst-case time complexity than matrix-multiplication time, at least when $i \in \Nbb$ is arbitrary.
Current options to speed up practical computations include sparsifying the filtered complex before computing homology~\cite{alonso-kerber-pritam}, as well as computational shortcuts that are known to significantly reduce computational time in practice~\cite{kerber-rolle,bauer-lenzen-lesnick}.
Nevertheless, computational cost is still a main bottleneck in real-world applications of persistence, limiting the size of the datasets on which it can be applied.


\paragraph{Zero-dimensional persistent homology}
In many applications, persistent homology in dimension zero is all that is required, as it encodes information about the changes in connectivity of filtered simplicial complexes, making it useful for clustering \cite{carlsson-memoli,carlsson-memoli-2,chazal-guibas-oudot-skraba,cai-kim-memoli-wang,rolle-scoccola,scoccola-rolle} and graph classification~\cite{zhao-wang,perslay,hofer-et-al,loiseaux-et-al,hacquard-lebovici}.

But, if one is only interested in $0$-dimensional homology $H_0(K, f; \kbb)$, algorithms relying on linear algebra are usually far from the most efficient ones:
For example, in the one-parameter case, the Betti tables of the $0$-dimensional homology of an $\Rbb$-filtered graph $(G,f)$ can be computed in $O(|G|\,\log|G|)$ time by first sorting the simplices of $G$ by their $f$-value, and then doing an ordered pass using a union-find data structure; in the language of barcodes, the Betti tables simply record the endpoints of the barcode, and the barcode can be computed using the elder rule (see \cite[pp.~188]{edelsbrunner-harer} or \cite[Algorithm~1]{rolle-scoccola}).

\paragraph{Contributions}
The paper is concerned with the following questions: What is the complexity of computing the Betti tables of graphs filtered by posets other than $\Rbb$?
To what extent do the Betti tables of filtered graphs depend on the field of coefficients?

\smallskip

We introduce algorithms for the computation of a minimal set of generators and of a minimal presentation of $0$-dimensional persistent homology indexed by an arbitrary poset.

\begin{restatable}{theoremx}{quadratictimealgo}
    \label{result:quadratic-time-algo}
    Let $\Pscr$ be any poset, and let $(G,f)$ be a finite $\Pscr$-filtered graph.
    \cref{algorithm:vertex-reduction} computes the $0$th Betti table of $H_0(G,f; \kbb)$ in $O\big(\,|G| \,\big)$ time.
    \cref{algorithm:minimal-presentation} computes a minimal presentation (and hence the $0$th and $1$st Betti tables) of $H_0(G,f; \kbb)$ in $O(\,|G|^2\,)$ time.
\end{restatable}

We also introduce a more efficient algorithm specialized to the two-parameter case, which computes all Betti tables, that is, $0$th, $1$st, and $2$nd, by Hilbert's syzygy theorem.

\begin{restatable}{theoremx}{loglineartime}
    \label{result:log-linear-time}
    Let $(G,f)$ be a finite $\Rbb^2$-filtered graph.
    \cref{algorithm:betti-tables}
    computes minimal presentations and the Betti tables of $H_0(G,f; \kbb)$ and $H_1(G,f; \kbb)$ in $O(\,|G|\, \log |G|\,)$ time.
\end{restatable}

Of note is the fact that Betti tables of $0$-dimensional two-parameter persistent homology can be computed in log-linear time, as in the one-parameter case.

\smallskip

As another main contribution, we establish a connection between minimal presentations and connectivity properties of filtered graphs (\cref{theorem:betti-tables-minimal-filtered-graph,theorem:second-betti-table}), which allows us to abstract away the algebraic problem, and focus on a simpler combinatorial problem.
The correctness of our algorithms is based on these results.

We conclude the paper by studying the dependence of Betti tables of zero-dimensional persistent homology on the field of coefficients.
Another interesting consequence of \cref{theorem:betti-tables-minimal-filtered-graph}
and the fact that the notion of minimal filtered graph is entirely combinatorial
is the following:



\begin{restatable}{theoremx}{mainresultindependencek}
    \label{result:main-result-independence-k}
    If $\Pscr$ is a poset and $(G,f)$ is a finite $\Pscr$-filtered graph, then the $0$th and $1$st Betti tables of $H_0(G,f; \kbb)$ are independent of the field $\kbb$.
\end{restatable}

For a reader familiar with one-parameter persistence, \cref{result:main-result-independence-k} may seem entirely natural, and perhaps not surprising; after all, the barcode of $0$-dimensional one-parameter persistent homology is independent of the field of coefficients (since the elder rule is independent of the field of coefficients).
But, remarkably, this intuition from the one-parameter case does not generalize: in the case of persistent modules indexed by non-linear posets, higher Betti tables of $0$-dimensional persistent homology can depend on the characteristic of the field of coefficients (and in fact, only on the characteristic), as our final main result shows.

\begin{restatable}{theoremx}{fielddependence}
    \label{result:field-dependence}
    For every $j \geq 2$, there exists a finite poset $\Pscr$ and a finite $\Pscr$-filtered graph $(G,f)$ such that $\beta_j(H_0(G,f;\kbb))$ depends on the characteristic of the field $\kbb$.
    However, if the characteristic of fields $\kbb$ and $\kbb'$ is the same, then $\beta_j(H_0(G,f;\kbb)) = \beta_j(H_0(G,f;\kbb'))$.
\end{restatable}

In particular, \cref{result:field-dependence} implies that \cref{result:main-result-independence-k} is tight, in the sense that the $0$th and $1$st Betti tables of $H_0(G,f;\kbb)$ are the only ones that are always independent of the field, when the indexing poset is arbitrary.
To prove \cref{result:field-dependence} we construct an example by hand for the case of $\beta_2(H_0(G,f;\kbb))$; in this example, we use the poset $\Pscr = \Rbb^4$, which in particular implies that field independence does not even hold in the familiar case of multiparameter persistence (i.e., when the indexing poset is $\Rbb^n$).
To prove \cref{result:field-dependence} for higher Betti tables, we use standard facts from the representation theory of finite dimensional algebras, and, in particular, a result of Igusa and Zacharia \cite[Theorem~1.2]{igusa-zacharia}, which relates Ext groups between simple poset representations to the reduced cohomology of the geometric realization of certain subposets.
\medskip

\noindent \textit{Implementation and experiments.}
Our final main contribution is the implementation of our main algorithms, and a computational comparison with RIVET~\cite{rivet} and mpfree~\cite{kerber-rolle} on several datasets, including synthetic 2D and 3D point clouds and real-world data from NYC taxi pickups.
Our experiments demonstrate that our algorithm is significantly faster and more memory-efficient than existing methods, particularly for large datasets. We also provide visualizations of Hilbert functions computed from Betti numbers, showcasing applications to clustering analysis.


\paragraph{Structure of the paper and summary of the approach}
This paper is an extended version of \cite{short-version}.
The main body of the paper has five sections: one on background (\cref{section:background}), one on theoretical results (\cref{section:theory}), one on algorithms (\cref{section:algorithms}), one on dependence of the field of coefficients (\cref{section:field-dependence}), and one on experiments (\cref{section:experiments}).
\cref{section:proofs} contains some technical proofs.

In order to describe and prove the correctness of our algorithms, we introduce, in the theory section, the notion of a \emph{minimal filtered graph} (Definition \ref{definition:minimal-filtered-graph}).
Informally, a minimal filtered graph is one whose vertices and edges induce a minimal presentation of its $0$-dimensional persistent homology (see \cref{theorem:betti-tables-minimal-filtered-graph}).
A graph is not minimal if it has some edge that can be either contracted or deleted without changing its $0$-dimensional persistent homology (for example, the graph of \cref{fig:graded-graph} has both contractible and deletable edges, and is thus not minimal; see Examples \ref{example:resolution-of-H_0} and \ref{example:minimal-vs-not-minimal}).
The idea is that, by contracting and deleting edges that do not change $0$-dimensional persistent homology, one inevitably ends up with a minimal filtered graph, from which a minimal presentation can be easily extracted.
Our main theoretical contributions (\cref{theorem:betti-tables-minimal-filtered-graph,theorem:second-betti-table}) make this idea precise by giving an explicit minimal presentation of $H_0$ of any minimal filtered graph, and an explicit minimal resolution of $H_0$ of any minimal $\Rbb^2$-filtered graph.
Our main algorithmic contributions are efficient algorithms for contracting and deleting edges as necessary.
\cref{algorithm:betti-tables} makes use of a dynamic tree data structure~\cite{sleator-tarjan}, which is the main ingredient that allows us to compute the Betti tables of bifilitered graphs in log-linear time;
we give more details in \cref{section:algorithms}.

\paragraph{Remark about multi-critical filtrations}
The filtrations considered in this paper are \emph{$1$-critical}, meaning that each simplex (vertex or edge, in the case of graphs) appears at exactly one grade.
In \cref{section:multi-critical}, we describe a simple preprocessing step to turn 
a finite multi-critical filtration by a lattice into a $1$-critical filtration with the same $0$-dimensional homology.
In the two-parameter case, this construction does not change the input complexity, but for other lattices it may increase the input size.
Note that the construction does not (necessarily) preserve $1$-dimensional homology.


\paragraph{Related work}
To the best of our knowledge, the only subcubic algorithm related to \cref{algorithm:betti-tables} is that of \cite{cai-kim-memoli-wang}, which, in particular, can be used to compute the Betti tables of $H_0$ of a function-Rips complex of a finite metric space $X$ in time $O(|X|^2\, \log |X|)$.
When applied to a function-Rips complex, our \cref{algorithm:betti-tables} has the same time complexity; however, our \cref{algorithm:betti-tables} applies to arbitrary bifiltered graphs, while function-Rips complexes are very special (and do not include arbitrary filtered graphs).
The inner workings of \cref{algorithm:betti-tables} are also different from those of~\cite{cai-kim-memoli-wang}: While we rely on dynamic trees, their algorithm relies on a dynamic minimum spanning tree, and, notably, on the fact that, in a function-Rips bifiltration, vertices are filtered exclusively by one of the two filtering functions (so that vertices are linearly ordered in the bifiltration).

The computation of minimal presentations of $\Rbb^n$-filtered complexes is studied in \cite{bender-gafvert-lesnick}.
Since they consider homology in all dimensions, their complexity is significantly worse than the quadratic complexity of \cref{algorithm:minimal-presentation}, which only applies to $0$-dimensional homology.


Part of the contributions in this paper could be rephrased using the language of discrete Morse theory for multiparameter filtrations \cite{allili-kaczynski-landi,scaramuccia,brouillette}, specifically, \cref{algorithm:local-components,algorithm:vertex-reduction} are essentially computing acyclic matchings which respect the filtration.
However, this point of view requires extra background, and, to the best of our understanding, it does not simplify the description of (the more interesting) \cref{algorithm:betti-tables}.

\section{Background}
\label{section:background}

As is common in persistence theory, we assume familiarity with very basic notions of category theory, specifically, that of a category and of a functor.
We let $\kbb$ denote a field, $\vect_\kbb$ denote the category of $\kbb$-vector spaces, and $\set$ denote the category of sets.
When the field $\kbb$ plays no role, we may denote $\vect_\kbb$ simply by $\vect$.

\paragraph{Graphs and filtered graphs}
A \emph{graph} $G = (V,E,\partial)$ consists of finite sets $V$ and $E$, and a function $\partial : E \to V \times V$.
We refer to the elements of $V$ as \emph{vertices}, typically denoted $v,w,x,y \in V$, and to the elements of $E$ as \emph{edges}, typically denoted $e,d,h \in E$.
If $\partial e = (v,w)$, we write $e_0 = v$ and $e_1 = w$.
The \emph{size} of a graph $G$ is $|G| = |V| + |E|$.

A \emph{subgraph} of a graph $G = (V,E,\partial)$ is a graph $G' = (V',E',\partial')$, where $V'\subseteq V$, $E' \subseteq E$, and such that $\partial' = \partial|_{E'}$ takes values in $V' \times V' \subseteq V \times V$.
If $G'$ is a subgraph of $G$, we write $G' \subseteq G$.

If $E' \subseteq E$ is a set of edges of $G$, we let $G \setminus E'$ be the subgraph of $G$ with the same vertices and $E \setminus E'$ as set of edges.

Let $\Pscr$ be a poset.
A \emph{$\Pscr$-filtered graph} $(G,f^V,f^E)$ consists of a graph $G$ and functions $f^V : V \to \Pscr$ and $f^E : E \to \Pscr$ such that $f^V(e_0) \leq f^E(e)$ and $f^V(e_1) \leq f^E(e)$ for all $e \in E$.
When there is no risk of confusion, we refer to $\Pscr$-filtered graphs simply as filtered graph, and denote both $f^V$ and $f^E$ by $f$ and the filtered graph $(G,f^V,f^E)$ by $(G,f)$.


If $(G,f)$ is a $\Pscr$-filtered graph and $r \in \Pscr$, we let $(G,f)_r$ be the subgraph of $G$ with vertices $\{v \in V : f(v) \leq r\}$ and edges $\{e \in E : f(e) \leq r\}$.


\paragraph{Persistence modules and persistent sets}
Let $\Pscr$ be a poset.
A \emph{$\Pscr$-persistence module} is a functor $\Pscr \to \vect$, where $\Pscr$ is the category associated with $\Pscr$.
Explicitly, a $\Pscr$-persistence module $M : \Pscr \to \vect$ consists of the following:
\begin{itemize}
    \item for each $r \in \Pscr$, a vector space $M(r)$;
    \item for each pair $r \leq s \in \Pscr$, a linear morphism $\phi^M_{r,s} : M(r) \to M(s)$; such that
    \item for all $r \in \Pscr$, the linear morphism $\phi^M_{r,r} : M(r) \to M(r)$ is the identity;
    \item for all $r \leq s \leq t \in \Pscr$, we have $\phi^M_{s,t} \circ \phi^M_{r,s} = \phi^M_{r,t} : M(r) \to M(t)$.
\end{itemize}
When there is no risk of confusion, we may refer to a $\Pscr$-persistence module as a persistence module.
An \emph{$n$-parameter persistence module} ($n \geq 1 \in \Nbb$) is an $\Rbb^n$-persistence module.

A \emph{morphism} $g : M \to N$ between persistence modules is a natural transformation between functors, that is, a family of linear maps $\{g_{r} : M(r) \to N(r)\}_{r \in \Pscr}$ with the property that
$\phi^N_{r,s} \circ g_r = g_s \circ \phi^M_{r,s} : M(r) \to N(s)$, for all $r \leq s \in \Pscr$.
Such a morphism is an \emph{isomorphism} if $g_r : M(r) \to N(r)$ is an isomorphism of vector spaces for all $r \in \Pscr$.

If $M,N : \Pscr \to \vect$ are persistence modules, their \emph{direct sum}, denoted $M \oplus N : \Pscr \to \vect$, is the persistence module with $(M\oplus N)(r) \coloneqq M(r) \oplus N(r)$ and with $\phi^{M\oplus N}_{r,s} \coloneqq  \phi^{M}_{r,s} \oplus \phi^{N}_{r,s} : M(r) \oplus N(r) \to M(s) \oplus N(s)$, for all $r \leq s \in \Pscr$.

\smallskip

Similarly, a \emph{$\Pscr$-persistent set} is a functor $\Pscr \to \set$.
The concepts of \emph{$n$-parameter persistent set}, and of \emph{morphism} and \emph{isomorphism} between persistent sets are defined analogously.

\paragraph{Persistent homology and connected components of filtered graphs}
\label{section:PH-graphs}
If $S \in \set$, we let $\langle S \rangle_\kbb \in \vect$ denote the free vector space generated by $S$; this defines a functor $\langle - \rangle_\kbb : \set \to \vect_\kbb$.
Let $G = (V,E,\partial)$ be a graph.
Consider the $\kbb$-linear map
\begin{align}
    \label{equation:definition-homology}
    \langle E\rangle_\kbb
      & \xrightarrow{\;\;d\;\;}
    \langle V \rangle_\kbb                         \\
    e & \xmapsto{\;\;\;\;\;\;} e_1 - e_0 \nonumber
\end{align}
The \emph{$0$-dimensional homology} of $G$, denoted $H_0(G; \kbb)$, is the $\kbb$-vector space $\coker(d)$, and the \emph{$1$-dimensional homology} of $G$, denoted $H_1(G; \kbb)$, is the $\kbb$-vector space $\ker(d)$.
In particular, every vertex $v \in V$ gives an element $[v] \in H_0(G; \kbb)$.
When there is no risk of confusion, we omit the field $\kbb$ and write $H_i(G)$ instead of $H_i(G;\kbb)$.

\smallskip

Homology is functorial, in the following sense.
If $G' = (V',E',\partial')$ is a subgraph of $G$, we have a commutative square
\[
    \begin{tikzpicture}
        \matrix (m) [matrix of math nodes,row sep=2em,column sep=3em,minimum width=2em,nodes={text height=1.75ex,text depth=0.25ex}]
        {
            \langle E'\rangle_\kbb & \langle V' \rangle_\kbb                         \\
            \langle E\rangle_\kbb  & \langle V \rangle_\kbb \\};
        \path[line width=0.75pt, -{>[width=8pt]}]
        (m-1-1) edge [above] node {$d'$} (m-1-2)
        (m-2-1) edge [above] node {$d$} (m-2-2)
        (m-1-1) edge [right hook-{>[width=8pt]}] (m-2-1)
        (m-1-2) edge [right hook-{>[width=8pt]}] (m-2-2)
        ;
    \end{tikzpicture}
\]
induced by the inclusions $V' \subseteq V$ and $E' \subseteq E$.
This induces linear maps $H_0(G') \to H_0(G)$ and $H_1(G') \to H_1(G)$.
Moreover, the morphism $H_\bullet(G'') \to H_\bullet(G)$ induced by a subgraph $G'' \subseteq G' \subseteq G$ is equal to the composite $H_\bullet(G'') \to H_\bullet(G') \to H_\bullet(G)$.

If $(G,f)$ is a $\Pscr$-filtered graph, and $i \in \{0,1\}$, we get a $\Pscr$-persistence module $H_i(G,f) : \Pscr \to \vect$, with $H_i(G,f)(r) = H_i((G,f)_r)$, and with structure morphism $H_i(G,f)(r) \to H_i(G,f)(s)$ for $r \leq s \in \Pscr$ induced by the inclusion of graphs $(G,f)_r \subseteq (G,f)_s$.

\smallskip

The set of \emph{connected components} of $G$, denoted $\pi_0(G)$, is the quotient of $V$ by the equivalence relation $\sim$ where $v \sim w \in V$ if and only if there exists a path in $G$ between $v$ and $w$.
If $v \in V$, we let $[v] \in \pi_0(G)$ denote its connected component, so that $[v] = [w] \in \pi_0(G)$ if and only if $v$ and $w$ belong to the same connected component.

The set of connected components is also functorial with respect to inclusions $G' \subseteq G$, since $[v] = [w] \in \pi_0(G')$ implies $[v] = [w] \in \pi_0(G)$.
In particular, if $(G,f)$ is a $\Pscr$-filtered graph, we get a $\Pscr$-persistent set $\pi_0(G,f) : \Pscr \to \set$, with $\pi_0(G,f)(r) = \pi_0((G,f)_r)$, and with the structure morphism $\pi_0(G,f)(r) \to \pi_0(G,f)(s)$ for $r \leq s \in \Pscr$ induced by the inclusion of graphs $(G,f)_r \subseteq (G,f)_s$.

\smallskip

The following is straightforward to check.

\begin{lemma}
    \label{lemma:H0-is-pi0-linearized}
    If $G$ is a graph, then the map $\langle \pi_0(G) \rangle_\kbb \to H_0(G)$ sending a basis element $[v] \in \pi_0(G)$ to $[v] \in H_0(G)$ is well-defined and an isomorphism of vector spaces.
    In particular, if $(G,f)$ is a $\Pscr$-filtered graph, composing the persistent set $\pi_0(G,f) : \Pscr \to \set$ with the free vector space functor $\langle - \rangle_\kbb : \set \to \vect$ yields a persistence module isomorphic to $H_0(G,f) : \Pscr \to \vect$.
    \qed
\end{lemma}

\paragraph{Projective persistence modules}
Given $r \in \Pscr$, let $\Psf_r : \Pscr \to \vect$ be the persistence module with $\Psf_r(s) = \kbb$ if $r \leq s$ and $\Psf_r(s) = 0$ if $r \nleq s$, with all structure morphisms $\kbb \to \kbb$ being the identity.
Equivalently, one can define $\Psf_r$ to be $H_0(\{x\},\emptyset,\partial,f)$, with $f(x) = r$.


\begin{notation}
    If $I$ is a finite set and $f : I \to \Pscr$ is any function, we can consider the direct sum $M = \bigoplus_{i \in I} \Psf_{f(i)} : \Pscr \to \vect$.
    When we need to work with elements of such a direct sum, we distinguish summands by writing $M = \bigoplus_{i \in I} \left(\Psf_{f(i)} \cdot \{i\}\right)$, 
    so that 
    $\left(\Psf_{f(i)} \cdot \{i\}\right)(r) = 0$ if $r \ngeq f(i)$ and $\left(\Psf_{f(i)} \cdot \{i\}\right)(r)$ is equal to the free vector space generated by $\{i\}$, for $r \geq f(i)$.
\end{notation}

\begin{definition}
    \label{definition:projective}
    A persistence module $M : \Pscr \to \vect$ is \emph{projective of finite rank} if there exists a function $\beta^M : \Pscr \to \Nbb$ of finite support such that $M \cong \bigoplus_{r \in \Pscr^n} \Psf_r^{\beta^M(r)}$.
\end{definition}

Note that, drawing inspiration from commutative algebra, projective persistence modules are sometimes also called \emph{free}.
The following result justifies the term projective used in Definition \ref{definition:projective}; see, e.g., \cite[Section~3.1]{rotman} for the usual notion of projective module.


\begin{lemma}
    \label{lemma:lifting}
    Let $g : M \to N$ be a surjection between $\Pscr$-persistence modules, and let $h : P \to N$ with $P$ projective of finite rank.
    There exists a morphism $h' : P \to M$ such that $g \circ h' = h$.
\end{lemma}
\begin{proof}
    Let $P = \bigoplus_{i \in I} \Psf_{f(i)} \cdot \{i\}$.
    To define a morphism $P \to M$ it is sufficient to say, for each $j \in J$, to which element of $M(f(j))$ the element $\{j\} \in P(f(j))$ gets mapped to.
    Since $g$ is a surjection, let us choose $m_j \in M(f(j))$ in the preimage of $h_{f(j)}(\{j\}) \in N(f(j))$, for each $j \in J$.
    Let $h' : P \to M$ be defined by mapping $\{j\} \in P(f(j))$ to $m_j \in M(f(j))$.
    It is clear that this morphism satisfies $g \circ h' = h$.
\end{proof}


\paragraph{Resolutions, presentations, and Betti tables}
The next result is standard; see, e.g., \cite[Proposition~6.24]{lesnick-course}.
%
%

\begin{lemma}
    \label{lemma:betti-table-characterization}
    If $M$ is projective of finite rank, then there exists exactly one function (necessarily of finite support) $\beta^M : \Pscr \to \Nbb$ such that $M \cong \bigoplus_{r \in \Pscr} \Psf_r^{\beta^M(r)}$.\qed
\end{lemma}

\begin{definition}
    The \emph{Betti table} of a persistence module $M : \Pscr \to \vect$
    that is projective of finite rank is the function $\beta^M : \Pscr \to \Nbb$ of Lemma \ref{lemma:betti-table-characterization}.
\end{definition}

The following notation is sometimes convenient.

\begin{notation}
    If $r \in \Pscr$, we let $\delta_r : \Pscr \to \Nbb$ be the function defined by $\delta_r(r) = 1$ and $\delta_r(s) = 0$ if $s \neq r$.
    In particular, $\delta_r = \beta^{\Psf_r} : \Pscr \to \Nbb$.
\end{notation}

\begin{definition}
    Let $M : \Pscr \to \vect$ be a $\Pscr$-persistence module.
    A \emph{finite projective cover} of $M$ is a surjective morphism $P \to M$ where $P$ is projective of finite rank, and $\sum_{r \in \Pscr} \beta^P(r) \in \Nbb$ is minimal.
\end{definition}



\begin{definition}
    Let $M : \Pscr \to \vect$ be a $\Pscr$-persistence module, and let $k \in \Nbb$.
    A \emph{finite projective $k$-resolution} (resp.~\emph{minimal finite projective $k$-resolution}) of $M$, denoted $C_\bullet \to M$, is a sequence of morphisms
    $C_k \xrightarrow{\partial_{k}} C_{k-1} \xrightarrow{\partial_{k-1}} \cdots \xrightarrow{\partial_2} C_1 \xrightarrow{\partial_1} C_0 \xrightarrow{\partial_0} M$ satisfying 
    \begin{itemize}
        \item $C_0 \to M$ is surjective (resp.~a projective cover);
        \item $\partial_i \circ \partial_{i+1} = 0$, for every $0 \leq i \leq k-1$ (so that $\partial_{i+1}$ factors through $\ker \partial_i$);
        \item $C_{i+1} \xrightarrow{\partial_{i+1}} \ker(\partial_i)$ is surjective (resp.~a projective cover), for every $0 \leq i \leq k-1$;
        \item $C_i$ is projective of finite rank, for every $0 \leq i \leq k$.
    \end{itemize}
\end{definition}

In particular, a minimal finite projective $0$-resolution is simply a projective cover.

\begin{notation}
    Since we only consider finite resolutions, we omit the word ``finite'' and simply say projective $k$-resolution.
    A \emph{(minimal) projective presentation} is a \emph{(minimal) projective $1$-resolution}.
\end{notation}

\begin{definition}
    Let $k \in \Nbb$.
    A persistence module $M : \Pscr \to \vect$ is \emph{finitely $k$-resolvable} if it admits a finite projective $k$-resolution.
\end{definition}


\begin{definition}
    \label{definition:betti-table}
    Let $M : \Pscr \to \vect$ be finitely $k$-resolvable and let $0 \leq i \leq k$.
    The \emph{$i$th Betti table} of $M$ is the function $\beta^M_i : \Pscr \to \Nbb$ (necessarily of finite support) defined as $\beta^M_i \coloneqq \beta^{C_i}$, where $C_\bullet \to M$ is a minimal $k$-resolution of $M$.
\end{definition}

\begin{notation}
When convenient, we write $\beta_i(M)$ instead of $\beta_i^M$ for the Betti tables of a persistence module $M$.
\end{notation}

The Betti tables of $M$ are also sometimes called the \emph{(multigraded) Betti numbers} of $M$.
The Betti tables of $M$, as defined in Definition \ref{definition:betti-table}, are independent of the choice of minimal presentation or resolution, thanks to the following standard result.


\begin{lemma}
    \label{lemma:minimal-resolution-uniqueness}
    Let $k \in \Nbb$.
    If $M : \Pscr \to \vect$ is finitely $k$-resolvable, then it admits a minimal projective $k$-resolution $C_\bullet \to M$.
    Moreover, any other projective $k$-resolution $C'_\bullet \to M$ has the property that $\beta^{C_i} \leq \beta^{C'_i}$ for all $0 \leq i \leq k$.
\end{lemma}
\begin{proof}
    We start by proving that any finitely $0$-resolvable $M$ admits a projective cover.
    By definition, there exists a surjection $P \to M$ from a projective of finite rank $P$.
    If $\sum_{r \in \Pscr}\beta^P(r)\in \Nbb$ is not minimal, we can remove at least one direct summand of $P$ to get $P' \to M$ still a surjection.
    Proceeding in this way, we necessarily end up with a projective cover of $M$.

    Now, a straightforward inductive argument shows that any $k$-resolvable module admits a minimal projective $k$-resolution.

%

    \smallskip

    We now prove that, if $g : P \to N$ is a projective cover and $h : Q \to N$ is a surjection from a finitely presentable projective, then $\beta^P \leq \beta^{Q}$.
    The inequality in the statement of the result involving Betti tables follows from this.

    By Lemma \ref{lemma:lifting}, there exist morphisms $g' : P \to Q$ and $h' : Q \to P$ such that $h \circ g' = g$ and $g \circ h' = h$.
    We claim that $h' \circ g' : P \to P$ is the identity; this is sufficient since then $Q \cong P \oplus P'$ for some finitely presentable projective $Q$, and thus $\beta^P \leq \beta^P + \beta^{P'} = \beta^{Q}$.

    Let $P = \bigoplus_{i \in I} \Psf_{f(i)} \cdot \{i\}$.
    Let $j \in J$.
    By Lemma \ref{lemma:projective-cover-minimal-elements}, we have that $(h' \circ g')_{f(j)}(\{j\}) = \{j\} \in P(f(j))$, since $g = g \circ (h' \circ g')$ and $\{j\}$ is the unique element of $P(f(j))$ mapping to $g_{f(j)}(\{j\}) \in M(f(j))$.
    Since the elements $\{j\} \in P(f(j))$ generate $P$, this implies that $h' \circ g'$ is the identity.
\end{proof}



\begin{example}
    \label{example:resolution-of-H_0}
    Consider the graph $G = (V,E,\partial)$ with $V = \{x,y\}$, $E = \{a,b\}$, and $\partial(a) = \partial(b) = (x,y)$.
    Consider the filtration $f : G \to \Rbb^2$ with $f(x) = f(y) = (0,0)$, $f(a) = (0,1)$, and $f(b) = (1,0)$.
    Then, a minimal resolution of $H_0(G,f)$ is given by
    \[
    0 \to \Psf_{(1,1)} \cdot \{\alpha\} \xrightarrow{\partial_2} \Psf_{(1,0)} \cdot \{a\} \oplus \Psf_{(0,1)} \cdot \{b\} \xrightarrow{\partial_1} \Psf_{(0,0)} \cdot \{x\} \oplus \Psf_{(0,0)} \cdot \{y\} \xrightarrow{\partial_0} H_0(G,f),
    \]
    where $\partial_0(\{x\}) = [x]$, $\partial_0(\{y\}) = [y]$, $\partial_1(\{a\}) = \partial_1(\{b\}) = \{y\} - \{x\}$, and $\partial_2(\{\alpha\}) = \{a\} - \{b\}$.
    In particular, $\beta_0(H_0(G,f)) = \delta_{(0,0)} + \delta_{(0,0)}$, $\beta_1(H_0(G,f)) = \delta_{(1,0)} + \delta_{(0,1)}$, and $\beta_2(H_0(G,f)) = \delta_{(1,1)}$.
    This can be easily checked by hand, but it also follows from \cref{theorem:betti-tables-minimal-filtered-graph,theorem:second-betti-table}, since $(G,f)$ is a minimal filtered graph, in the sense of Definition \ref{definition:minimal-filtered-graph}, which we give in the next section.
\end{example}

%
%
%

\section{Theory}
\label{section:theory}

We start with a simple, standard result which says that a projective presentation of $0$-dimensional homology is given by using the grades of vertices as generators, and the grades of edges as relations.

\begin{lemma}
    \label{definition:H0-H1}
    Let $\Pscr$ be a poset, let $(G,f)$ be a $\Pscr$-filtered graph, and consider the following morphism of projective modules
    \begin{align*}
        \label{equation:standard-presentation}
        \bigoplus_{e \in E} \Psf_{f(e)} \cdot \{e\}
        \;\;\;       & \xrightarrow{\;\;\;\partial_1^{(G,f)}\;\;\;}\;\;\;
        \bigoplus_{v \in V} \Psf_{f(v)} \cdot \{v\}                                  \\
        \{e\} \;\;\; & \xmapsto{\;\;\;\;\;\;\;\;\;\;\;\;\;\;\;\;} \;\;\;\{e_1\} - \{e_0\}\nonumber
    \end{align*}
    Then $H_0(G,f) \cong \coker\left(\partial_1^{(G,f)}\right)$ and $H_1(G,f) \cong \ker\left(\partial_1^{(G,f)}\right)$.
    In particular, $H_0(G,f)$ is a finitely presentable persistence module.
\end{lemma}
\begin{proof}
    It is straightforward to see that the evaluation of the natural transformation $\partial_1^{(G,f)}$ at $r \in \Pscr$ coincides with morphism Equation \ref{equation:definition-homology}, where the graph is $\left(\{v \in V : f(v) \leq r\}, \{e \in E : f(e) \leq r\}, \partial\right)$.
    Then $H_0(G,f)$ is finitely presentable by definition.
\end{proof}

The point of minimal filtered graphs, which we now introduce, is that
they make the presentation in Lemma \ref{definition:H0-H1} minimal (see \cref{theorem:betti-tables-minimal-filtered-graph}).

\begin{definition}
    \label{definition:minimal-filtered-graph}
    Let $(V, E, \partial,f)$ be a filtered graph and let $e \in E$.
    \begin{itemize}
        \item The edge $e$ is \emph{collapsible} if $e_0 \neq e_1$ and $f(e) = f(e_i)$ for some $i \in \{0,1\}$.
        \item The edge $e$ is \emph{deletable} if $[e_0] = [e_1] \in \pi_0(V, E\setminus \{e\}, \partial, f)(f(e))$.
    \end{itemize}
    A filtered graph $(G,f)$ is
    \begin{itemize}
        \item \emph{vertex-minimal} if it does not contain any collapsible edges.
        \item \emph{edge-minimal} if it does not contain any deletable edges.
        \item \emph{minimal} if it is vertex-minimal and edge-minimal.
    \end{itemize}
\end{definition}

\begin{example}
    \label{example:minimal-vs-not-minimal}
    The graph of Example \ref{example:resolution-of-H_0} is minimal, as neither of the two edges is collapsible or deletable.
    On the other hand, the graph of \cref{fig:graded-graph} is not minimal since $h_1$, $h_2$, $d_1$, $d_2$, and $d_3$ are collapsible, and $e_3$ is deletable.
\end{example}

As their name suggests, collapsible and deletable edges can be collapsed or deleted:

\begin{construction}
    \label{construction:collapse-deletion}
    Let $(G,f) = (V, E, \partial,f)$ be a filtered graph.
    \begin{itemize}
        \item If $e \in E$ is collapsible, let $i \in \{0,1\}$ be such that $f(e) = f(e_i)$.
        Define the \emph{simple collapse} $G\downarrow e \coloneqq (V\setminus \{e_i\}, E \setminus \{e\}, \partial')$, where $\partial' = (\phi \times \phi) \circ \partial$ and $\phi : V \to V$ is given by $\phi(v) = v$ if $v \neq e_i$ and $\phi(v) = e_{1-i}$ if $v = e_i$.
        \item If $e \in E$ is any edge, define the \emph{simple deletion} $G\setminus e \coloneqq (V, E\setminus \{e\}, \partial)$.
    \end{itemize}
\end{construction}

We now study the effect of collapsing collapsible edges, and of deleting deletable edges.
We start with a useful observation, whose proof is straightforward.

\begin{lemma}
    \label{lemma:deleting-collapsin-does-not-change}
    Let $(G,f)$ be a filtered graph and let $e$ and $d$ be two different edges of $G$.
    \begin{enumerate}
        \item Assume that $e$ is collapsible.
        If $d$ is not collapsible (resp.~deletable) in $(G,f)$, then $d$ is not collapsible (resp.~deletable) in $(G\downarrow e, f)$.
        \item If $d$ is not collapsible (resp.~deletable) in $(G,f)$, then $d$ is not collapsible (resp.~deletable) in $(G\setminus e, f)$.\qed
    \end{enumerate}
\end{lemma}


\begin{lemma}
    \label{lemma:vertex-reduction-same-homology}
    If $(G,f) = (V,E,\partial,f)$ is a filtered graph and $e \in E$ is collapsible, then $H_i(G,f) \cong H_i(G\downarrow e,f)$ for $i \in \{0,1\}$.
\end{lemma}
\begin{proof}
    This follows at once from the homotopy-invariance property of homology; however, in order to make the exposition self-contained, let us give an algebraic proof, using the definition of homology for graphs given in \cref{section:PH-graphs}.

    Let $e \in E$ be collapsible, and assume that $f(e) = f(e_0)$; the case $f(e) = f(e_1)$ is analogous.
    Consider the following diagram
    \[
        \small
    \begin{tikzpicture}
        \matrix (m) [matrix of math nodes,row sep=8em,column sep=10em,minimum width=2em,nodes={text height=1.75ex,text depth=0.25ex}]
        {
        \displaystyle \left(\bigoplus_{d \in E \setminus \{e\}} \Psf_{f(d)} \cdot \{d\}\right) \oplus \left(\Psf_{f(e)} \cdot \{e\}\right) &
        \displaystyle \left(\bigoplus_{v \in V \setminus \{e_i\}} \Psf_{f(v)} \cdot \{v\}\right) \oplus \left(\Psf_{f(e_i)} \cdot \{e_0\} \right)\\
        \displaystyle \bigoplus_{d \in E} \Psf_{f(d)} \cdot \{d\} &
        \displaystyle \bigoplus_{v \in V} \Psf_{f(v)} \cdot \{v\}\;.\\};
        \path[line width=0.75pt, -{>[width=8pt]}]
        (m-1-1) edge [above] node {
            \scriptsize
            $\left[
            \begin{array}{c|c}
                \partial_1^{(G\downarrow e, f)} & 0\\ \hline
                0 & 1
            \end{array}
            \right]$
            } (m-1-2)
        (m-2-1) edge [below] node {$\partial_1^{(G,f)}$} (m-2-2)
        ($(m-1-1)-(0,0.7)$) edge [right] node {
            \scriptsize
            $\left[
            \begin{array}{c}
                \{d\}\mapsto
                \begin{cases}
                    \{d\} & \text{if $d_1 \neq e_0$}\\
                    \{d\} + \{e\} & \text{else}
                \end{cases}
                 \\\hline
                \{e\} \mapsto \{e\} 
            \end{array}
            \right]^T$
        } (m-2-1)
        ($(m-1-2)-(0,0.7)$) edge [left] node {
            \scriptsize
            $\left[
            \begin{array}{c}
                \{v\}\mapsto \{v\} \\\hline
                \{e_0\} \mapsto \{e_1\}-\{e_0\} 
            \end{array}
            \right]^T$
        } (m-2-2)
        ;
    \end{tikzpicture}
    \]
    It is clear that the cokernel (resp.~kernel) of the top horizontal morphism is isomorphic to the cokernel (resp.~kernel) of $\partial_1^{(G\downarrow e, f)}$, which in turn is isomorphic to the $0$th (resp.~$1$st) persistent homology of $(G\downarrow e, f)$, by Lemma \ref{definition:H0-H1}.
    So it is sufficient to prove that the diagram is commutative, and that the vertical morphisms are isomorphisms.
    The commutativity of the diagram is a straightforward check, using the definition of $G\downarrow e$ (\cref{construction:collapse-deletion}) in the case where $d \in E$ is such that $d_1 = e_0$.
    The fact that the vertical morphisms are isomorphisms is proven by a simple Gaussian elimination using the matrix representation of the vertical morphisms.
\end{proof}

\begin{lemma}
    \label{lemma:edge-reduction-H1}
    If $(G,f) = (V,E,\partial,f)$ is a filtered graph and $e \in E$ is deletable, then
    $H_0(G,f) \cong H_0(G \setminus e,f)$ and $H_1(G,f) \cong H_1(G \setminus e ,f) \oplus \Psf_{f(e)}$.
\end{lemma}
\begin{proof}
    Lemma \ref{lemma:deletable-pi-0-iso}, together with Lemma \ref{lemma:H0-is-pi0-linearized}, implies that the morphism $H_0(G,f) \to H_0(G \setminus e,f)$ induced by the inclusion of graphs is an isomorphism of persistence modules.

    We now prove that $H_1(G,f) \cong H_1(G \setminus e ,f) \oplus \Psf_{f(e)}$.
    Consider the following diagram:
    \[
    \begin{tikzpicture}
        \matrix (m) [matrix of math nodes,row sep=4em,column sep=3em,nodes={asymmetrical rectangle,text height=1.75ex,text depth=0.25ex}]
        {
            0
            & H_1(G\setminus e ,f)
            & H_1(G,f)
            & \Psf_{f(e)} \\
             0
            & \displaystyle \bigoplus_{d \in E \setminus \{e\}} \Psf_{f(d)}
            & \displaystyle \bigoplus_{d \in E} \Psf_{f(d)}
            & \Psf_{f(e)}
            & 0\\
             0
            & \displaystyle \bigoplus_{v \in V} \Psf_{f(v)}
            & \displaystyle \bigoplus_{v \in V} \Psf_{f(v)}
            & 0
            & 0\\
            & H_0(G \setminus e, f)
            & H_0(G, f)
            & 0 \\};
        \path[line width=0.75pt, -{>[width=8pt]}]
        (m-1-1) edge (m-1-2)
        (m-1-2) edge (m-1-3)
        (m-1-3) edge (m-1-4)
        (m-2-1) edge (m-2-2)
        (m-2-2) edge (m-2-3)
        (m-2-3) edge (m-2-4)
        (m-2-4) edge (m-2-5)
        (m-3-1) edge (m-3-2)
        (m-3-2) edge (m-3-3)
        (m-3-3) edge (m-3-4)
        (m-3-4) edge (m-3-5)
        (m-4-2) edge [above] node {$\cong$} (m-4-3)
        (m-4-3) edge (m-4-4)
        (m-1-2) edge (m-2-2)
        ($(m-2-2)-(0,0.8)$) edge [right] node {$\partial_1^{(G\setminus e,f)}$} (m-3-2)
        (m-3-2) edge (m-4-2)
        (m-1-3) edge (m-2-3)
        ($(m-2-3)-(0,0.8)$) edge  [right] node {$\partial_1^{(G,f)}$} (m-3-3)
        (m-3-3) edge (m-4-3)
        (m-1-4) edge (m-2-4)
        (m-2-4) edge (m-3-4)
        (m-3-4) edge (m-4-4)
        ;
        \draw[-{>[width=8pt]},rounded corners,dashed]
        (m-1-4) -| ($(m-2-5)+(0.5,0)$) |- ($(m-2-3)-(0,0.7)$) -| ($(m-3-1) - (0.5,0)$) |- (m-4-2);
    \end{tikzpicture}
    \]
    The first and last modules in the second and third column are, respectively, a kernel and a cokernel, by Lemma \ref{definition:H0-H1}, and the second and third row are clearly exact.
    Thus, the snake lemma \cite[Corollary~6.12]{rotman} implies the existence of the dashed morphism and the exactness of the sequence of modules which includes the dashed morphism.
    Moreover, the statement about $H_0$ proven in the first part of the proof implies the isomorphism in the bottom row.
    We thus have that $0 \to H_1(G\setminus e, f) \to H_1(G,f ) \to \Psf_{f(e)} \to 0$ is a short exact sequence, which necessarily splits since $\Psf_{f(e)}$ is projective \cite[Proposition~3.3]{rotman}.
    The result follows.
\end{proof}

Next is a construction of a minimal presentation of $H_0$ of a minimal filtered graph.


\begin{theorem}
    \label{theorem:betti-tables-minimal-filtered-graph}
    Let $\Pscr$ be a poset, let $(G,f)$ is a $\Pscr$-filtered graph, and let $\partial_1^{(G,f)} : C_1 \to C_0$ be the morphism of Lemma \ref{definition:H0-H1}.
    \begin{enumerate}
        \item If $(G,f)$ is vertex-minimal, then the map $C_0 \to \coker(\partial_1^{(G,f)})$ is a projective cover.
        In this case, $\beta_0\big(H_0(G,f)\big) = \sum_{v \in V} \delta_{f(v)}$.
        \item If $(G,f)$ is a minimal filtered graph, then $C_1 \xrightarrow{\partial_1^{(G,f)}} C_0 \to \coker(\partial_1^{(G,f)})$ is a minimal presentation.
        In this case, we also have 
        $\beta_1\big(H_0(G,f)\big) = \sum_{e \in E} \delta_{f(e)}$.
    \end{enumerate}
\end{theorem}
\begin{proof}
    The claim about the Betti tables follows directly from the claim that the presentation $\partial_1^{(G,f)}$ is minimal, so let us prove this claim.

    We start by proving (1).
    Let $G = (V,E,\partial)$ be vertex-minimal.
    We first prove that
    $h : Q \coloneqq \bigoplus_{v \in V} \Psf_{f(v)} \cdot \{v\} \to \coker{\partial_1^{(G,f)}} \cong H_0(G,f) \cong \langle \pi_0(G,f)(f(v))\rangle_\kbb$ is a projective cover, where in the last isomorphism we used Lemma \ref{lemma:H0-is-pi0-linearized}.
    This morphism maps $\{v\} \in Q(f(v))$ to $[v] \in \langle \pi_0(G,f)(f(v))\rangle_\kbb$, for every $v \in V$.
    If the morphism were not a projective cover, then there would exist $v' \in V$ such that the restriction
    $h' : \bigoplus_{v \in V \setminus \{v'\}} \Psf_{f(v)} \cdot \{v\} \to \coker{\partial_1^{(G,f)}} \cong H_0(G,f) \cong \langle \pi_0(G,f)(f(v))\rangle_\kbb$ would still be surjective.
    This is not possible: $(G,f)$ is vertex-minimal, so $[v]\neq [v'] \in \pi_0(G,f)(f(v'))$ for every $v \in V\setminus \{v'\}$, and thus $h'$ would not be surjective, since $[v'] \in \langle \pi_0(G,f)(f(v))\rangle_\kbb$ would not be in its image.

    The claim about $\beta_0$ follows at once from the definition of Betti table (Definition \ref{definition:betti-table}).

    \medskip

    We now prove (2).
    Let $G = (V,E,\partial)$ be a minimal filtered graph.
    We prove that $g : \bigoplus_{e \in E} \Psf_{f(e)} \cdot \{e\} \to \ker(h)$ is a projective cover.
    If it were not, then there would exist $e' \in E$ such that the restriction $g' : P \coloneqq \bigoplus_{e \in E \setminus \{e'\}} \Psf_{f(e)} \cdot \{e\} \to \ker(h)$ of $g$ would still be surjective.
    Let $G' = (V, E\setminus \{e\}, \partial)$.
    In particular, the last two vertical morphisms in the following commutative diagram would be isomorphisms
    \[
        \begin{tikzpicture}
            \matrix (m) [matrix of math nodes,row sep=3em,column sep=3em,minimum width=2em,nodes={text height=1.75ex,text depth=0.25ex}]
            {\displaystyle
                \bigoplus_{e \in E \setminus \{e'\}} \Psf_{f(e)} \cdot \{e\}
                 & \displaystyle\bigoplus_{v \in V} \Psf_{f(v)} \cdot \{v\}
                 & H_0(G',f) & \langle \pi_0(G',f)\rangle_\kbb \\
                \displaystyle\bigoplus_{e \in E} \Psf_{f(e)} \cdot \{e\}
                 & \displaystyle\bigoplus_{v \in V} \Psf_{f(v)} \cdot \{v\} & H_0(G,f) & \langle \pi_0(G,f)\rangle_\kbb\, . \\};
            \path[line width=0.75pt, -{>[width=8pt]}]
            (m-1-1) edge [above] node {$\partial_1^{(G',f)}$} (m-1-2)
            (m-2-1) edge [above] node {$\partial_1^{(G,f)}$} (m-2-2)
            (m-1-2) edge (m-1-3)
            (m-1-3) edge [above] node {$\cong$} (m-1-4)
            (m-2-2) edge (m-2-3)
            (m-2-3) edge [above] node {$\cong$} (m-2-4)
            (m-1-1) edge [right hook-{>[width=8pt]}] (m-2-1)
            (m-1-2) edge [double equal sign distance,-] (m-2-2)
            (m-1-3) edge (m-2-3)
            (m-1-4) edge (m-2-4)
            ;
        \end{tikzpicture}
    \]
    This is not possible, since the last vertical morphism maps $[e'_0], [e'_1]\in \langle  \pi_0(G',f)(f(e')) \rangle_\kbb$ to
    $[e'_0], [e'_1]\in \langle  \pi_0(G,f)(f(e')) \rangle_\kbb$, and,
    $[e'_0] = [e'_1]\in \langle  \pi_0(G,f)(f(e')) \rangle_\kbb$, since $e'$ is an edge of $G$, while $[e'_0] \neq [e'_1] \in \langle  \pi_0(G',f)(f(e')) \rangle_\kbb$, by the fact that $(G,f)$ is edge-minimal.
    
    As for the case of (1), the claim about $\beta_1$ follows from the definition of Betti table.
\end{proof}

We now focus on the case of $\Rbb^2$-filtered graphs.
If $(G,f)$ is an $\Rbb^2$-filtered graph, let $f_{\mathbf x},f_{\mathbf y} : G \to \Rbb$ denote the first and second coordinates of $f$, respectively.

\begin{definition}
    \label{definition:cycle-creating}
    Let $(V,E,\partial,f)$ be a minimal $\Rbb^2$-filtered graph, and let $\prec$ be a total order on the set of edges $E$ that refines the lexicographic order induced by $f$ (i.e., $e \prec e'$ implies that $f_{\mathbf x}(e) < f_{\mathbf x}(e')$ or $f_{\mathbf x}(e) = f_{\mathbf x}(e')$ and $f_{\mathbf y}(e) \leq f_{\mathbf y}(e')$).
    An edge $d \in E$ is \emph{cycle-creating} with respect to $\prec$ if $[d_0] = [d_1] \in \pi_0(V, \{e \in E : e \prec d\}, \partial)$.
\end{definition}

Thus, $d \in E$ is cycle-creating if there exists a list of edges $e^1, \dots, e^k \in E$ with $e^i \prec d$ for $1 \leq i \leq k$, and a list of signs $s^1, \dots, s^k \in \{+,-\}$, such that $s^1 e^1, \dots, s^k e^k$ is a directed path from $d_0$ to $d_1$.
Such a path $w = (s^\bullet, e^\bullet)$ is called a \emph{witness} for $d$, and we denote $f(w) \coloneqq \left(\,f_{\mathbf x}(d)\,,\, \max_{1 \leq i \leq k} f_{\mathbf y}(e^i) \, \right)$.
A \emph{minimal witness} $w$ of a cycle-creating edge $d$ is one for which $f_{\mathbf y}(w) \in \Rbb$ is as small as possible.

Since $\prec$ refines the lexicographic order, we have $f_{\mathbf x}(e^i) \leq f_{\mathbf x}(d)$ for all $i$; thus, if $\max_{1 \leq i \leq k} f_{\mathbf y}(e^i) \leq f_{\mathbf{y}}(d)$, we would have that $[d_0] = [d_1] \in \pi_0(V,E \setminus \{d\},\partial,f)(f(d))$, which contradicts the fact that the filtered graph is minimal.
This means that if $w$ is a witness for~$d$,
then $\max_{1 \leq i \leq k} f_{\mathbf y}(e^i) > f_{\mathbf{y}}(d)$, and thus $f(d) < f(w)$.

\begin{theorem}
    \label{theorem:second-betti-table}
    Let $(G,f) = (V,E,\partial, f)$ be a minimal $\Rbb^2$-filtered graph, and $\prec$ a total order on $E$ refining the lexicographic order.
    For each $d \in E$ cycle-creating, let $w_d = (s_d^\bullet, e_d^\bullet)$ be a minimal witness wrt to $\prec$.
    The kernel of the morphism $\partial_1^{(G,f)}$ of Lemma \ref{definition:H0-H1} is given by:
    \begin{align*} \bigoplus_{\substack{d \in E\\\text{cycle-creating}}}
        \Psf_{f(w_d)} \cdot \{d\}
        \;\;\;       & \xrightarrow{\;\;\kappa^{(G,f)}\;\;}\;\;\;
        \bigoplus_{e \in E} \Psf_{f(e)} \cdot \{e\}                                                         \\
        \{d\} \;\;\; & \xmapsto{\;\;\;\;\;\;\;\;\;\;\;\;\;\;} \;\;\;\{d\} - \big(s_d^1 \{e_d^1\} + \cdots + s_d^k\{e_d^k\}\big).
    \end{align*}
    It follows that
    \begin{itemize}
        \item $\beta_0(H_0(G,f)) = \sum_{v \in V} \delta_{f(v)}$;
        \item $\beta_1(H_0(G,f)) = \sum_{e \in E} \delta_{f(e)}$;
        \item $\beta_0(H_1(G,f)) = \beta_2(H_0(G,f)) = \sum_{d \in E, \text{cycle-creating}} \delta_{f(w_d)}$;
        \item $\beta_i(H_0(G,f)) = 0$ for $i \geq 3$, and $\beta_i(H_1(G,f)) = 0$ for $i \geq 1$.
    \end{itemize}
\end{theorem}
\begin{proof}
    The claim about the Betti tables follows at once from the first statement of the theorem together with \cref{theorem:betti-tables-minimal-filtered-graph}.
    Note in particular that we need not prove that $\kappa^{(G,f)}$ and $\partial_1^{(G,f)}$ form a minimal resolution, since we know that $\partial_1^{(G,f)}$ is a minimal presentation.
    So let us prove the claim about the kernel.

    We proceed by induction on the number of edges in the graph.
    Let $h \in E$ be the maximum of $E$ with respect to the total order $\prec$.
    The filtered graph $(G',f) \coloneqq (G\setminus h, f)$ is minimal, by Lemma \ref{lemma:deleting-collapsin-does-not-change}(2), so the result holds for this graph.

    We consider the following diagram, obtained by adding the edge $h$ to $(G',f)$:
    \[
    \scriptsize
    \hspace*{-5px}
        \begin{tikzpicture}
            \matrix (m) [matrix of math nodes,row sep=4em,column sep=3em,minimum width=2em,nodes={text height=1.75ex,text depth=0.25ex}]
            {
            0   & \displaystyle \bigoplus_{\text{\emph{cc}}\, d \in E \setminus \{h\}} \Psf_{f(w_d)} \cdot \{d\}
                & \displaystyle \bigoplus_{e \in E \setminus \{h\}} \Psf_{f(e)} \cdot \{e\}
                & \displaystyle\bigoplus_{v \in V} \Psf_{f(v)} \cdot \{v\}
                & H_0(G',f) & 0 \\
            0   & \displaystyle \bigoplus_{\text{\emph{cc}}\, d \in E} \Psf_{f(w_d)} \cdot \{d\}
                & \displaystyle\bigoplus_{e \in E} \Psf_{f(e)} \cdot \{e\}
                & \displaystyle\bigoplus_{v \in V} \Psf_{f(v)} \cdot \{v\} & H_0(G,f) & 0\, , \\};
            \path[line width=0.75pt, -{>[width=8pt]}]
            (m-1-1) edge (m-1-2)
            (m-1-2) edge [above] node {$\kappa^{(G',f)}$} (m-1-3)
            (m-2-1) edge (m-2-2)
            (m-2-2) edge [above] node {$\kappa^{(G,f)}$} (m-2-3)
            (m-1-3) edge [above] node {$\partial_1^{(G',f)}$} (m-1-4)
            (m-2-3) edge [above] node {$\partial_1^{(G,f)}$} (m-2-4)
            (m-1-4) edge (m-1-5)
            (m-2-4) edge (m-2-5)
            (m-1-3) edge [right hook-{>[width=8pt]}] (m-2-3)
            (m-1-2) edge [right hook-{>[width=8pt]}] (m-2-2)
            (m-1-4) edge [double equal sign distance,-] (m-2-4)
            (m-1-5) edge [->>] (m-2-5)
            (m-1-5) edge (m-1-6)
            (m-2-5) edge (m-2-6)
            ;
        \end{tikzpicture}
    \]
    where \emph{cc} stands for ``cycle-creating''.
    Since $h$ is the maximum according to $\prec$, the cycle-creating edges of $(G',f)$ with respect to $\prec$ are simply the cycle-creating edges of $(G,f)$, possibly minus $h$ in the case where $h$ is cycle-creating.

    The top row of the diagram is exact by inductive hypothesis, and we need to prove that the bottom row is exact as well.

    \medskip

\noindent \begin{minipage}{0.75\textwidth}
    \hspace{12px} In order to prove this, we evaluate the persistence modules and morphisms of the diagram at each $(x,y) \in \Rbb^2$, and show that the bottom row is exact.
    We consider three cases:
    \begin{enumerate}
        \item $(x,y) \not\geq f(h)$;
        \item $(x,y) \geq f(h)$, and either $h$ is not cycle-creating or $(x,y) \not\geq f(w_h)$;
        \item $(x,y) \geq f(w_h)$, which is only considered if $h$ is cycle-creating.
    \end{enumerate}
\end{minipage}
\begin{minipage}{0.2\textwidth}
    \hspace*{10pt} \includegraphics[width=\textwidth]{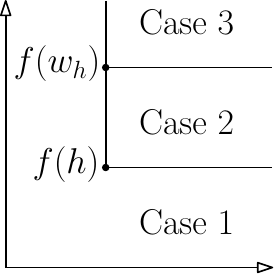}
\end{minipage}

    \medskip

    Case (1) is immediate, since the top and bottom row are equal when evaluated at $(x,y)$, as $h$ is not an edge of $(G,f)_{(x,y)}$.
    
    \smallskip

    Case (2) is slightly more interesting.
    In this case, the difference between the two rows, when evaluated at $(x,y)$, is that the edge $h$ has appeared in the second row, and it is connecting the connected components of $h_0$ and $h_1$, which would otherwise be distinct.
    This is because either $h$ is not cycle-creating, or $(x,y) \not\geq f(w_h)$, so that $y < f_{\mathbf y}(w_h)$, which is the smallest grade at which $h_0$ and $h_1$ become connected in $\{(x',y') \in \Rbb^2 : (x',y') \geq f(h)\}$, by the fact that $w_h$ is a minimal witness.
    In going from the top row to the bottom row, both evaluated at $(x,y)$, the dimension of $H_0$ is dropping by one, and the dimension of the vector space generated by edges is going up by one.
    Since, in the bottom row, the generator $\{h\}$ is mapping to $\{h_1\}-\{h_0\}$, which is non-zero in $H_0(G',f)(x,y)$, this case follows from the following basic claim from linear algebra, whose proof is straightforward.

    \smallskip

    \noindent\textit{Claim.}
    Let $0 \to \ker \alpha \to A \xrightarrow{\,\,\alpha\,\,} B \to \coker \alpha \to 0$ be an exact sequence of vector spaces, and let $b \neq 0 \in B$.
    Then, the following is an exact sequence of vector spaces
    $0 \to \ker \alpha \to A \oplus \kbb \xrightarrow{(\alpha,b)} B \to \coker (\alpha,b) \to 0$.

    \smallskip

    Finally, we consider case (3).
    Thus, in going from the first to the second row, when evaluated at $(x,y)$, the dimension of the vector space generated by edges is going up by one, and the dimension of the space generated by vertices remains the same.
    In this case, in $H_0(G',f)(x,y)$, the elements $[h_0]$ and $[h_1]$ are already equal, since there is a path between $h_0$ and $h_1$, namely, the witness $w_h$; thus, the dimension of $\ker \partial_1(x,y)$ is going up by one.
    In order to account for this, $\kappa^{(G,f)}$ is mapping $\{h\}$ to $\{h\} - (s^1_h\{e^1_d\} + \cdots + s^k_h\{e^k_d\})$, which is clearly non-zero, and in the kernel of $\partial_1(x,y)$.
    This concludes the proof.
\end{proof}

\section{Algorithms}
\label{section:algorithms}

\subsection{Outline and background}

We first introduce some technical notions and overview the algorithms.

Two $\Pscr$-filtered graph $(G,f)$ and $(G',f')$ are \emph{homology-equivalent} (over~$\kbb$) if $H_0(G,f;\kbb) \cong H_0(G',f';\kbb)$ and $H_1(G,f;\kbb) \cong H_1(G',f';\kbb)$.

\cref{algorithm:vertex-reduction} reduces the input filtered graph to a vertex-minimal filtered graph by collapsing edges; it relies on \cref{algorithm:local-components}, which collapses local collapsible edges.
A \emph{local collapsible edge} of $(G,f)$ is an edge $e$ such that $e_0 \neq e_1$ and $f(e) = f(e_0) = f(e_1)$.
\cref{algorithm:vertex-reduction} then identifies \emph{minimal vertices}, that is, vertices with no adjacent collapsible edge decreasing the grade, and runs a depth-first search from each of these to build and collapse a tree of collapsible edges.

\cref{algorithm:minimal-presentation} first calls \cref{algorithm:vertex-reduction} to perform all collapses and then deletes deletable edges until there are no more deletable edges.
It then builds the presentation of Lemma \ref{definition:H0-H1}.

\cref{algorithm:betti-tables} first calls \cref{algorithm:vertex-reduction} to perform all collapses, and then goes through all edges, with respect to the lex order on their grades, and identifies deletable edges, non-deletable edges, and cycle-creating edges.

\paragraph{Dynamic dendrograms}
Since this is used in \cref{algorithm:betti-tables}, we now introduce the dynamic dendrogram data structure, which, informally,
represents a dendrogram where elements merge as time goes from $-\infty$ to $\infty$, and is dynamic in that one is allowed to change the dendrogram by making two elements merge earlier.


\begin{definition}
    Let $(G,f) = (V,E,\partial,f)$ be a $[-\infty,\infty)$-filtered graph.
    A \emph{dynamic dendrogram} $D$ for $(G,f)$ is a data structure supporting the following operations:
    \begin{itemize}
        \item If $v,w \in V$, the operation $D.\timeofmerge(v,w)$ returns the smallest $r \in [-\infty,\infty)$ such that $[v] = [w] \in \pi_0(G,f)(r)$, or $\infty$ if $[v] \neq [w] \in \pi_0(G,f)(r)$ for all $r \in [-\infty,\infty)$.
        \item If $v,w \in V$ and $t \in [-\infty, \infty)$, the operation $D.\mergeattime(v,w,t)$ modifies the dendrogram so that it is a dynamic dendrogram for $(V, E \sqcup \{e\}, \partial', f')$, with $f'|_{E} = f$, $\partial'|_{E} = \partial$, $f'(e) = t$, and $\partial'(e) = (v,w)$.
    \end{itemize}
\end{definition}

A dynamic dendrogram $D$ can easily and efficiently be implemented using a \emph{mergeable tree} $T$ in the sense of~\cite{georgiadis}.
These trees store heap-ordered forests, i.e., a collection of rooted labeled trees, where the labels decrease (or in our case increase) along paths to the root. The data structures support a wealth of operations for dynamic updates. We need only three of them: $\mathsf{insert}(v,t)$ inserts node $v$ with label $t$ into the forest; $\mathsf{nca}(v,w)$ finds the nearest common ancestor of two nodes $v$ and $w$; $\mathsf{merge}(v,w)$ merges the paths of $v$ and $w$ to their root(s) while preserving the heap order.
We use these operations to implement dynamic dendrograms as follows:
\begin{itemize}
    \item To implement $D.\timeofmerge(v,w)$, return the label of $T.\mathsf{nca}(v,w)$.
    \item To implement $D.\mergeattime(v,w,t)$, let $h$ be a new vertex not already in the mergeable tree $T$, do $T.\mathsf{insert}(h, t)$, then $T.\mathsf{merge}(v,h)$, and $T.\mathsf{merge}(w,h)$.
\end{itemize}



\paragraph{Presentation matrices}
In the algorithms, to represent a Betti table $\beta_i(M)$ we use a list of elements of the indexing poset.
If we have represented the $0$th and $1$st Betti tables of $M$ with lists $\beta_0$ and $\beta_1$, we represent a minimal presentation for $M$ with a sparse matrix using coordinate format, that is, with a list of triples $(i,j,v)$ representing the fact that $v$ is the entry at row $i$ and column $j$, and where $i$ represents the $i$th element of $\beta_0$ and $j$ represents the $j$th element of $\beta_1$.

\medskip

\begin{algorithm}[H]
    \caption{Collapse local collapsible edges}
    \label{algorithm:local-components}
    \SetAlgoLined
    \DontPrintSemicolon
    \SetKwInput{Input}{Input}
    \SetKwInput{Output}{Output}
    \Input{Filtered graph $(G,f) = (V,E,\partial,f)$}
    \Output{Filtered graph $(G',f)$ homology-equivalent to $(G,f)$ and without local collapsible edges}
    Initialize dictionary $\phi$ with identity map, $\phi[v] = v$ for $v \in V$ \;
    Initialize empty set $\var{visited}$ \;
    Initialize $E' \gets E$ \;
    \For(\Comment{Run depth-first search from $v$ on local edges}){$v \in V$}{
        Initialize $\var{stack}$ with $(v,\emptyset,v)$ \;
        \While{$\var{stack}$ is not empty}{
            $(v,e,u) \gets \var{stack}.\op{pop}()$ \;
            \If{$u \notin \var{visited}$}{
                add $u$ to $\var{visited}$ \;
                \If{$e \neq \emptyset$}{
                    $\phi[u] \gets v$ \;
                    remove $e$ from $E'$ \;
                }
                \For(\Comment{local collapsible}){$e \in E$ with $\{e_0,e_1\} = \{u,x\}$ and $f(e) = f(x) = f(u)$}{
                    \op{push} $(v,e,x)$ onto $\var{stack}$ \;
                }
            }
        }
    }
    $V' \gets \{v \in V : \phi[v] = v \}$ \;
    \Return $(G',f) = \left(V',E',(\phi \times \phi) \circ \partial, f\right)$
\end{algorithm}

\begin{algorithm}[H]
    \caption{Collapse to vertex-minimal filtered graph}
    \label{algorithm:vertex-reduction}
    \SetAlgoLined
    \DontPrintSemicolon
    \SetKwInput{Input}{Input}
    \SetKwInput{Output}{Output}
    \Input{Filtered graph $(G,f) = (V,E,\partial,f)$}
    \Output{
        Vertex-minimal filtered graph $(G',f)$ homology-equivalent to $(G,f)$, and $\beta_0(H_0(G,f))$
    }
    $(G_i,f) = (V_i,E_i,\partial_i,f) \gets \textrm{\cref{algorithm:local-components}}(G,f)$ \;
    Initialize dictionary $\phi$ with identity map, $\phi[v] = v$ for $v \in V_i$ \;
    Initialize empty set $\var{visited}$ \;
    Initialize $E' \gets E_i$ \;
    \For(\Comment{Run depth-first search from minimal vertices}){$v$ in $V_i$}{
        \If(\Comment{$v$ is not minimal}){$\exists$ edge $e \in E_i$ with $\{e_0,e_1\} = \{u,v\}$ and $f(e) = f(v) > f(u)$}{
            \op{continue} \;
        }
        Initialize $\var{stack}$ with $(v, \emptyset, v)$ \;
        \While{$\var{stack}$ is not empty}{
            $v,e,u \gets \var{stack}.\op{pop}()$ \;
            \If{$u \notin \var{visited}$}{
                add $u$ to $\var{visited}$ \;
                \If{$e \neq \emptyset$}{
                    $\phi[u] \gets v$ \;
                    remove $e$ from $E'$ \;
                }
                \For{every edge $e \in E_i$ with $\{e_0,e_1\} = \{u,x\}$ with $f(e) = f(x) > f(u)$}{
                    \op{push} $(v,e,x)$ onto $\var{stack}$ \;
                }
            }
        }
    }
    $V' \gets \{v \in V_i : \phi[v] = v \}$ \;
    \Return $(G',f) = \left(V', E', (\phi \times \phi) \circ \partial_i, f\right)$
    and $\beta_0 = [f(v) : v \in V'] $ \;
\end{algorithm}

\begin{algorithm}[H]
    \caption{Minimal presentation of $\Pscr$-filtered graph}
    \label{algorithm:minimal-presentation}
    \SetAlgoLined
    \DontPrintSemicolon
    \SetKwInput{Input}{Input}
    \SetKwInput{Output}{Output}
    \Input{Filtered graph $(G,f) = (V,E,\partial,f)$}
    \Output{Minimal presentation of $H_0(G,f;\kbb)$}
    Initialize $\beta_0, \beta_1$ with empty lists\;
    Initialize empty sparse matrix $\Mcal$ and empty dictionary $\mathsf{row\_idx}$\;
    $(V',E',\partial',f) \gets \textrm{\cref{algorithm:vertex-reduction}}(G,f)$\;
    \For(\Comment{Check each edge, and delete it if it is deletable}){$e \in E'$ \label{line:for-deletable}}{
        Define set of vertices
        $V_e \gets \{v \in V' : f(v) \leq f(e)\}$\;
        Define set of edges 
        $E_e \gets \{d \in E' : f(d) \leq f(e)\} \setminus e$\;
        Run breadth-first search on $(V_e,E_e,\partial')$ starting from $e_0$\;
        \If(\Comment{$[e_0] = [e_1] \in \pi_0(V_e,E_e,\partial')$ so $e$ is deletable}){$e_1$ is reachable from $e_0$}{
            $E' \gets E' \setminus \{e\}$\;
        }
    }
    \For{$v \in V'$}{
        $\beta_0.\append(f(v))$\;
        $\mathsf{row\_idx}[v] \gets |\beta_0|$\;
    }
    \For(\Comment{The morphism $\partial_1^{(V',E',\partial')}}$ of Lemma \ref{definition:H0-H1}){$e \in E'$}{
        $\beta_1.\append(f(e))$\;
        $\Mcal.\append(\,\mathsf{row\_idx}[e_0]\,,\,|\beta_1|\,,\,-1\,)$\;
        $\Mcal.\append(\,\mathsf{row\_idx}[e_1]\,,\,|\beta_1|\,,\,1\,)$\;
    }
    \Return $\beta_0, \beta_1, \Mcal$
\end{algorithm}

\begin{algorithm}[H]
    \caption{Betti tables and minimal presentation of $\Rbb^2$-filtered graph}
    \label{algorithm:betti-tables}
    \SetAlgoLined
    \DontPrintSemicolon
    \SetKwInput{Input}{Input}
    \SetKwInput{Output}{Output}
    \Input{Filtered graph $(G,f) = (V,E,\partial,f)$}
    \Output{Betti tables $\beta_0,\beta_1,\beta_2$ and minimal presentation of $H_0(G,f;\kbb)$, and $\beta_0^1 \coloneqq \beta_0(H_1(G,f;\kbb))$}
    Initialize $\beta_0, \beta_1, \beta_2, \beta^1_0$ with empty lists\;
    Initialize empty sparse matrix $\Mcal$ and empty dictionary $\mathsf{row\_idx}$\;
    $(V',E',\partial',f) \gets \textrm{\cref{algorithm:vertex-reduction}}(G,f)$\;
    Let $D$ be a dynamic dendrogram on $(V', \emptyset, \partial, g)$, with $g(v) = -\infty$ for all $v \in V'$
    \;
    \For(\Comment{Visit grades lexicographically}){$(x,y) \in \Rbb^2$ s.t.~$f^{-1}(x,y) \neq \emptyset$ in lex order}{
    \For(\Comment{All vertices belong to the projective cover}){$v \in V' \text{ with } f(v) = (x,y)$
    }{
    $\beta_0.\append(\,(x,y)\,)$\;
    $\mathsf{row\_idx}[v] \gets |\beta_0|$\;
    }
    \For{$e \in E' \text{ with } f(e) = (x,y)$ \label{line:for-edges}}{
    $s \gets D.\timeofmerge(e_0,e_1)$\;
    \If(\Comment{The edge is deletable, so it only affects $H_1$}){$s \leq y$ \label{line:if-s-y}}{
        $\beta^1_0.\append(\,(x,y)\,)$\;
    }
    \Else(\Comment{Edge is not deletable, so belongs to relations in resolution}){
        $D.\mergeattime(e_0,e_1,y)$\;
        $\beta_1.\append(\,(x,y)\,)$\;
        $\Mcal.\append(\,\mathsf{row\_idx}[e_0]\,,\,|\beta_1|\,,\,-1\,)$\;
        $\Mcal.\append(\,\mathsf{row\_idx}[e_1]\,,\,|\beta_1|\,,\,1\,)$\;
        \If(\Comment{The edge is cycle-creating}){$s < \infty$}{
            $\beta_2.\append(\,(x,s)\,)$\; \label{line:betti-2-added}
            $\beta^1_0.\append(\,(x,s)\,)$\;
        }
    }
    }
    }
        \Return $\beta_0, \beta_1, \beta_2, \beta^1_0, \Mcal$
\end{algorithm}


\subsection{Complexity and correctness}
We conclude the paper by using the theoretical results of \cref{section:theory} to prove the main results in the introduction.

We start by with a convenient lemma, which gives conditions under which one can collapse an entire subgraph and produce a homology-equivalent filtered graph.
The following definition describes the type of subgraph that can be collapsed.

\begin{definition}
    \label{definition:monotonic-forest}
    Let $(G,f) = (V,E,\partial, f)$ be a filtered graph.
    A subset $E' \subseteq E$ is a \emph{monotonic forest} of $(G,f)$ if:
    \begin{enumerate}
        \item The subgraph of $G$ spanned by the edges in $E'$ is a forest.
        \item For every vertex $v$ in the forest, there exists at most one edge $e \in E'$ such that $\{e_0,e_1\} = \{v,w\}$ and $f(w) < f(v)$.
    \end{enumerate}
\end{definition}

\begin{lemma}
    \label{lemma:collapsible-forest}
    Let $(G,f) = (V,E,\partial,f)$ be a filtered graph, and let $E' \subseteq E$ be a set of collapsible edges, which forms a monotonic forest.
    If $e \in E'$, then all the edges in $E'\setminus\{e\}$ are collapsible in $(G\downarrow e, f)$.
    In particular, the whole forest $E'$ can be collapsed (in any order) to obtain a filtered graph that is homology-equivalent to $(G,f)$.
\end{lemma}
\begin{proof}
    Collapsing a single collapsible edge produces a homology-equivalent filtered graph, by Lemma \ref{lemma:vertex-reduction-same-homology}.
    So we need to prove that collapsing an edge of the monotonic forest $E'$ consisting of collapsible edges does not affect the collapsibility of other edges in $E'$.

    Let $e,d \in E'$.
    We prove that $d$ is collapsible in $(G\downarrow e, f)$.
    Without loss of generality, we assume that $f(d) = f(d_1)$.
    Without loss of generality, we also assume that $f(e) = f(e_1)$ and that the simple collapse deletes $e_1$ and $e$, and makes every edge adjacent to $e_1$ now be adjacent to $e_0$.

    Let $G' = (V', E\setminus \{e\}, \partial') = G\downarrow e$.
    For notational convenience, let $(d'_0, d'_1) = \partial'(d)$.
    We prove that $d'_0 \neq d'_1$ and $f(d'_0) = f(d)$ or $f(d'_1) = f(d)$, which implies that $d$ is collapsible in $(G\downarrow d, f)$, by definition.
    The fact that $d'_0 \neq d'_1$ in $G\downarrow e$ is clear, since $d_0 \neq d_1$, so $d'_0 = d'_1$ would mean that $\{e_0, e_1\} = \{d_0,d_1\}$, which is not possible because $E'$ is a forest and thus contains no loops.

    We prove $f(d'_0) = f(d)$ or $f(d'_1) = f(d)$ by considering two cases.
    The first case is when we have $d_1 \neq e_1$.
    In this case $d'_1 = d_1$, so $f(d'_1) = f(d_1) = f(d)$.
    The second case is when $d_1 = e_1$.
    In this case, we cannot have $f(d_0) < f(d)$, as this would violate condition (2) of Definition \ref{definition:monotonic-forest} for the vertex $v = e_1$, so we must have $f(d_0) = f(d)$.
    But since $d_0 \neq d_1 = e_1$, we have that $d'_0 = d_0$, so $f(d'_0) = f(d_0) = f(d)$.
\end{proof}

\begin{lemma}
    \label{lemma:algo-1-correct}
    Let $(G,f)$ be a $\Pscr$-filtered graph.
    \cref{algorithm:local-components} outputs
    a filtered graph homology-equivalent to $(G,f)$ and without local collapsible edges in time $O(|G|)$.
\end{lemma}
\begin{proof}
    The algorithm performs a depth-first search on the local collapsible edges.
    It removes the edges of the depth-first search forest and makes each edge adjacent to a vertex in the forest now be adjacent to the root of a tree in the forest, using the map $\phi$.
    This is equivalent to collapsing the forest, which produces a homology-equivalent graph, by Lemma \ref{lemma:collapsible-forest}.
    To see that the output contains no local collapsible edges, note that all local collapsible edges in the initial graph are identified by the algorithm, and collapsing an edge cannot make any previously non-collapsible edge be collapsible, by Lemma \ref{lemma:deleting-collapsin-does-not-change}(1).
    Moreover, since collapsing a local collapsible edge does not change the $f$-value of the endpoints of any edge, a local collapse cannot make a previously non-local, collapsible edge become local collapsible.

    The time complexity of the depth-first search is linear.
\end{proof}

\begin{proposition}
    \label{proposition:vertex-reduction-correctness}
    Let $(G,f)$ be a $\Pscr$-filtered graph.
    \cref{algorithm:vertex-reduction} outputs a vertex-minimal filtered graph homology-equivalent to $(G,f)$ and $\beta_0(H_0(G,f))$ in $O(|G|)$ time.
\end{proposition}
\begin{proof}
    The guiding principle for collapsing collapsible edges followed by the algorithm is simple: first identify a maximal monotonic tree (Definition \ref{definition:monotonic-forest}) of collapsible edges; then collapse that subtree, which does not change the homology by Lemma \ref{lemma:collapsible-forest}; and then iterate this.
    In order to do this efficiently, note that any maximal monotonic tree of collapsible edges contains exactly one local connected component of the graph (that is, a connected component in the subgraph of vertices and edges having the same $f$-value) that is minimal (that is, that has no adjacent collapsible edge decreasing the grade).
    Since the input is assumed to not have any locally collapsible edges,
    minimal local components are the same as minimal vertices,
    which are easy to identify, since these are simply vertices with no adjacent collapsible edge decreasing the grade.

    The algorithm iterates over minimal vertices, and collapses the corresponding DFS tree of collapsible edges going up in the filtration.
    Since collapsing edges cannot produce more collapsible edges
    (Lemma \ref{lemma:deleting-collapsin-does-not-change}(1)), nor more minimal vertices, at the end of the main for-loop there are no collapsible edges left.

    The output Betti table $\beta_0$ is correct by \cref{theorem:betti-tables-minimal-filtered-graph}(1), and the time complexity of the depth-first search is linear.
\end{proof}

\begin{proposition}
    \label{proposition:edge-deletion-correctness}
    Let $(G,f)$ be a finite $\Pscr$-filtered graph.
    \cref{algorithm:minimal-presentation} outputs a minimal presentation of $H_0(G,f)$ in $O(\,|G|^2\,)$ time.
\end{proposition}
\begin{proof}
    We start with the claim about the time complexity.
    \cref{algorithm:vertex-reduction} is linear by Proposition \ref{proposition:vertex-reduction-correctness}, so it remains to prove that we can check the condition $[e_0] = [e_1] \in \pi_0(V_e,E_e,\partial')$ in $O(|G|)$ time.
    To do this, we run breadth-first search on $G'$, starting from $e_0$, and simply ignore every vertex and every edge whose $f$-value is not less than or equal to $f(e)$; and we also ignore $e$.
    If at some point we encounter $e_1$, then $[e_0] = [e_1] \in \pi_0(V_e,E_e,\partial')$, and otherwise $[e_0] \neq [e_1] \in \pi_0(V_e,E_e,\partial')$.

    We now prove the correctness of the algorithm.
    By \cref{theorem:betti-tables-minimal-filtered-graph}, we only need to check
    that, at the end of the for-loop in Line \ref{line:for-deletable}, the graph $(V',E',\partial',f)$ contains no deletable edges.
    Clearly, the for-loop visits all edges, so correctness follows from
    Lemma \ref{lemma:deleting-collapsin-does-not-change}(2).
\end{proof}

\quadratictimealgo*
\begin{proof}
    This follows directly from Propositions \ref{proposition:vertex-reduction-correctness} and \ref{proposition:edge-deletion-correctness}.
\end{proof}

\loglineartime*
\begin{proof}
    The output $\beta_0$ is correct by the correctness of \cref{algorithm:vertex-reduction}.
    To prove correctness of the rest of the algorithm, we observe the following:
    by iterating over edges one-by-one in a way that respects the lexicographic order, we are in fact iterating over edges according to a total order $\prec$ on $E'$ that refines the lexicographic order induced by $f$.
    We consider the following invariant:
    \begin{itemize}[leftmargin=20pt]
        \item[$(\ast)$] At the end of the iteration of the main for-loop in Line \ref{line:for-edges} corresponding to the edge $e \in E'$, the dynamic dendrogram $D$ is a dynamic dendrogram for the filtered graph $(V',\{e' \in E' : e' \preceq e\},\partial',g)$,
        where $g(v) \coloneqq -\infty$ for all $v \in V'$, and $g(e') \coloneqq f_{\mathbf y}(e')$.
    \end{itemize}
    It is clear that the algorithm maintains invariant $(\ast)$, since $\prec$ is, by definition, the total order in which we visit the edges.

    We now address the correctness of the outputs $\beta_0^1$, $\beta_1$, and $\Mcal$.
    An element of $\beta_0^1$ is added either in the if-clause in Line
    \ref{line:if-s-y} checking for $s \leq y$, or when an element of $\beta_2$ is added in Line \ref{line:betti-2-added}.
    The if-clause corresponds exactly to an edge that is connecting two vertices in the same connected component, that is, to a deletable edge.
    Thus, by Theorem \ref{theorem:second-betti-table} and Lemma \ref{lemma:edge-reduction-H1}, the output $\beta_0^1$ is correct as long as the output for $\beta_2$ is.
    It also follows that the output $\beta_1$ is correct, as deletable edges are being identified correctly, and that the output presentation matrix $\Mcal$ is correct, by \cref{theorem:betti-tables-minimal-filtered-graph}.

    To prove that the output $\beta_2$ is correct, we need to show that cycle-creating edges are being identified correctly.
    To see this, note that the check $s < \infty$ is being performed only after knowing that the edge $e$ is not deletable (and thus after knowing that the edge will be part of the minimal graph, by Lemma \ref{lemma:deleting-collapsin-does-not-change}), and that the if-clause $s < \infty$ is true exactly when $e_0$ and $e_1$ belong to the same connected component at some grade $(x',y')$ with $x' < f_{\mathbf x}(e)$, by the invariant $(\ast)$.
    This is precisely the case when $e$ is cycle-creating, by Definition \ref{definition:cycle-creating}, and we have $s = f_{\mathbf y}(w_e)$, where $w_e$ is a minimal witness of $e$, also by the invariant $(\ast)$.
    If follows that the output $\beta_2$ is correct.

    To conclude, we analyze the time complexity.
    Since the algorithm consists of for-loops visiting each vertex and each edge exactly once, we only need to make sure that the operations involving $D$ are in $\log(|G|)$ which is proven in \cite[Section~4]{georgiadis}.
\end{proof}

\subsection{Multi-critical filtrations}
\label{section:multi-critical}

If $X$ is a set, let $\mathsf{Parts}_f(X)$ denote the set of finite, non-empty subsets of $X$.
A \emph{finite multi-critical filtration} of a graph $G = (V,E)$ by a poset $\Pscr$ consists of functions $f_V : V \to \mathsf{Parts}_f(\Pscr)$  and $f_E : E \to \mathsf{Parts}_f(\Pscr)$ such that, for every $e \in E$ and $p \in f_E(e)$ there exists $q_0 \in f_V(e_0)$ and $q_1 \in f_V(e_1)$ with $q_0,q_1 \leq p$.

Let $\Pscr$ be a lattice, and let $(G,f)$ be a finite multi-critical filtration.
There exist algorithms for reducing the computation of the homology of multi-critical filtrations to that of critical filtrations \cite{chacholski-scolamiero-vaccarino}, but these cannot, a priori, be directly combined with our algorithms.
We thus describe a custom procedure for constructing a ($1$-critical) filtered graph with the same $0$-dimensional persistent homology as $(G,f)$.

In words, we proceed as follows:
for each vertex in the original filtration, we add a new copy of the vertex at all the grades at which the vertex is born and an edge between copies of the same vertex at all the minimal joins of the birth grades;
and for each edge in the original filtration, we add a new copy of the edge (connecting any one copy of its endpoints) at all the grades at which the edge is born.
More precisely, we proceed as follows.
As set of vertices, we take
\[
    V' \coloneqq \{(v, p) : v \in V, p \in f_V(v)\}.
\]
In order to identify all the different copies of a vertex $v$, we first define
\[
    E''_v \coloneqq \{(v, p,p') : p \neq p' \in f_V(v)\}\; \text{ with }\; \partial'(v,p,p') = ((v,p), (v,p')).
\]
Consider the function $\phi : E''_v \to \Pscr$ given by $\phi(v,p,p') = p \vee p'$.
We then let $E'_v \subseteq E''_v$ be such that the cardinality of $E'_v$ is minimal and $\phi(E'_v)$ contains all the minimal elements of~$\phi(E''_v)$.
The idea is that, in order to identify the different copies of a vertex $v$, it is sufficient to use the edges in $E'_v$.
To account for the edges of $G$, let
\[
    E'_{edg} \coloneqq \{(e,p) : e \in E, p \in f_E(e)\} \; \text{ and } \; \partial'(e,p) = ((e_0,q_0), (e_1,q_1)),
\]
where $q_0 \in f_V(e_0)$ and $q_1 \in f_V(e_1)$ satisfy $q_0,q_1 \leq p$.
As set of edges, we take $E' \coloneqq \left(\bigcup_{v \in V} E'_v\right) \cup E'_{edg}$, and define the filtration as follows
\[
    f_{V'}(v,p) = p, \;\; f_{E'}(v,p,p') \coloneqq p \vee p'\,\,,\, \text{ and }\; f_{E'}(e,p) \coloneqq p.
\]
We leave it to the reader to check that the $0$-dimensional persistent homology of the $1$-critical filtered graph $(V', E', \partial', f_{V'}, f_{E'})$ is isomorphic to that of $(G,f)$.
One way to see this, is to show that $\pi_0(V', E', \partial', f_{V'}, f_{E'}) \cong \pi_0(G,f)$ using the continuous function from the geometric realization of $(V', E', \partial')$ to that of $(V,E)$ defined by $|(v,p)| \mapsto |v|$, $|(v,p,p')| \mapsto |v|$, $|(e,p)| \mapsto |e|$.

\section{Dependence on the field}
\label{section:field-dependence}

\subsection{Field-independence}

The following is a standard characterization of minimal resolutions.

\begin{lemma}
    \label{lemma:characterization-minimal-resolution}
    Let $M : \Pscr \to \vect$ be a $\Pscr$-persistence module, and let $C_\bullet \to M$ be a finite resolution.
    The resolution $C_\bullet$ is minimal if and only if, for every $a \in \Pscr$ and every $j \in \Nbb$, we have that, if $\Psf_a$ is a direct summand of both $C_{j+1}$ and $C_{j}$, then the differential $C_{j+1} \to C_{j}$ composed with any isomorphisms $C_{j+1} \cong C_{j+1}' \oplus \Psf_a$ and $C_j \cong C_j' \oplus \Psf_a$ restricts to the zero morphism $\Psf_a \to \Psf_a$.
\end{lemma}
\begin{proof}
    If the restriction to $\Psf_a \to \Psf_a$ is not zero, then, using the lifting property of $\Psf_a$ (Lemma \ref{lemma:lifting}), the resolution $C_\bullet$ can be decomposed as a direct sum of a resolution $C'_\bullet$ and a \emph{trivial resolution}, i.e., one which is zero everywhere except on degrees $j$ and $j+1$, where it is $\Psf_a$, with the differential being the identity.
    In that case, the resolution would not be minimal.

    If the resolution is not minimal, then it decomposes as a direct sum of a resolution $C'_\bullet$ and a trivial resolution \cite[Theorem~11.26]{lesnick-course}, which then implies that there exist $a \in \Pscr$, $j \in \Nbb$, and isomorphisms $C_{j+1} \cong C_{j+1}' \oplus \Psf_a$ and $C_j \cong C_j' \oplus \Psf_a$ such that the restriction $\Psf_a \to \Psf_a$ is non-zero.
\end{proof}

\begin{proposition}
    \label{proposition:same-characteristic-same-betti}
    Let $\Pscr$ be a poset and let $(G,f)$ be a finite $\Pscr$-filtered graph.
    If $\kbb \subseteq \kbb'$ is a field extension, and $H_0(G,f;\kbb)$ is finitely resolvable, then $H_0(G,f;\kbb')$ is finitely resolvable and $\beta_j(H_0(G,f;\kbb)) = \beta_j(H_0(G,f;\kbb'))$ for every $j \in \Nbb$.
\end{proposition}
\begin{proof}
    By the universal coefficient theorem, we have $H_0(G,f;\kbb') = H_0(G,f;\kbb) \otimes_{\kbb} \kbb'$, and since tensoring by a field extension is exact, any minimal resolution of $H_0(G,f;\kbb)$ tensored with $\kbb'$ gives a resolution of $H_0(G,f;\kbb')$.
    To conclude, we must show that tensoring by a field extension preserves minimality of resolutions, which follows directly from Lemma \ref{lemma:characterization-minimal-resolution}.
\end{proof}

It follows from Proposition~\ref{proposition:same-characteristic-same-betti} that the Betti tables of zero-dimensional persistent homology can only depend on the characteristic of the field.

%

\mainresultindependencek*
\begin{proof}
    The statement holds for minimal filtered graphs, by \cref{theorem:betti-tables-minimal-filtered-graph}.
    If the graph is not minimal, one can iteratively collapse edges and then delete edges to make the filtered graph minimal.
    The notions of collapsible and deletable edge, and of simple collapse and simple deletion, do not depend on the field, so the result for non-minimal graphs follows from Lemmas \ref{lemma:vertex-reduction-same-homology} and \ref{lemma:edge-reduction-H1}, which imply that the effect on Betti tables of collapses and deletions is also independent of the field.
\end{proof}

\subsection{Field-dependence of second Betti table}

The key idea is summarized in the following construction and lemma.

Consider the finite poset $\Pscr$ with underlying set $\{a,b,c,d\}$ with $a$, $b$, and $c$ incomparable, and all of them smaller than $d$.
Consider the following diagram $M : \Pscr \to \Ab$ of Abelian groups indexed by $\Pscr$:

\[
    \begin{tikzpicture}
        \matrix (m) [matrix of math nodes,row sep=3em,column sep=3em,minimum width=2em,nodes={text height=1.75ex,text depth=0.25ex}]
        {
            \Zbb & \Zbb^3 & \Zbb \\
            \, & \Zbb & \, \\};
        \path[line width=0.75pt, -{>[width=8pt]}]

        (m-1-1) edge [above] node
        {\scalebox{0.6}{
            $\begin{bmatrix}
            0 \\ 1 \\ 1
        \end{bmatrix}$}}
        (m-1-2)

        (m-2-2) edge [left] node
        {\scalebox{0.6}{
            $\begin{bmatrix}
            1 \\ 0 \\ 1
        \end{bmatrix}$}} (m-1-2)

        (m-1-3) edge [above] node
        {\scalebox{0.6}{
            $\begin{bmatrix}
            1 \\ 1 \\ 0
        \end{bmatrix}$}} (m-1-2)
        ;
    \end{tikzpicture}
\]

By tensoring the Abelian groups in the above diagram with a field $\kbb$ we obtain a $\Pscr$-persistence module, which we denote by $M_\kbb : \Pscr \to \vect_\kbb$.

\begin{lemma}
    \label{lemma:simple-example-field-dependence}
    If the characteristic of $\kbb$ is not $2$, then a minimal projective presentation of $M_\kbb$ is given by $0 \to \Psf_a \oplus \Psf_b \oplus \Psf_c$.
    If the characteristic of $\kbb$ is $2$, then a minimal projective presentation of $M_\kbb$ is given by $(1, 1, 1, 0)^T : \Psf_d \to \Psf_a \oplus \Psf_b \oplus \Psf_c \oplus \Psf_d$.
    In particular, the zeroth and first Betti tables of $M_\kbb$ depend on the characteristic of $\kbb$.
\end{lemma}
\begin{proof}
    Regardless of the field $\kbb$, the codomain of a projective presentation of $M_\kbb$ must have $\Psf_a$, $\Psf_b$, and $\Psf_c$ as direct summands, since $M_\kbb(a) = M_\kbb(b) = M_\kbb(c) = \kbb \neq 0$, and $a$, $b$, and $c$ are all minimal elements of the poset.
    If the characteristic of $\kbb$ is different from $2$, then $M_\kbb(d)$ is generated by the images of the structure morphisms corresponding to $a \leq d$, $b \leq d$, and $c \leq d$, so $M_\kbb$ is projective and the minimal projective presentation is as in the statement.
    
    If the characteristic of $\kbb$ is $2$, the images of the structure morphisms corresponding to $a \leq d$, $b \leq d$, and $c \leq d$ only generate a $2$-dimensional subspace of $M_\kbb(d)$, so the codomain of any projective presentation must also have $\Psf_d$ as a direct summand, and it is straightforward to check that the minimal projective presentation is as in the statement.
\end{proof}

In order to use the idea above to produce an $\Rbb^4$-filtered graph with the property that the second Betti of its zero-dimensional persistent homology depends on the field, we first define a diagram of Abelian groups $N : \Rbb^4 \to \Ab$ as the left Kan-extension of $M$ along the inclusion $\iota : \Pscr \to \Rbb^4$ mapping $a\mapsto (0,1,1,1)$, $b \mapsto (1,0,1,1)$, $c \mapsto (1,1,0,1)$, and $d \mapsto (1,1,1,1)$.
As before, we define $N_\kbb : \Rbb^4 \to \vect_\kbb$ by tensoring $N$ with $\kbb$.

The same proof as that of \cref{lemma:simple-example-field-dependence} shows the following.
\begin{lemma}
    \label{lemma:R4-example-field-dependence}
    The zeroth and first Betti tables of $N_\kbb$ depend on the characteristic of $\kbb$.\qed
\end{lemma}

We now construct an $\Rbb^4$-filtered graph whose second Betti table depends on
the characteristic of the field; see \cref{fig:graph-field-characteristic}.
Let $G = (V,E,\partial)$ with $V = \{x,y\}$, $E = \{r,s,t,u\}$, and $\partial : E \to V \times V$ constantly $(x,y)$.
Define $f : G \to \Rbb^4$ by
\begin{align*}
    x &\mapsto (0,0,0,0),\;
    y \mapsto (0,0,0,0),\\
    r \mapsto (0,0,1,1),\;
    s &\mapsto (0,1,0,1),\;
    t \mapsto (1,0,0,1),\;
    u \mapsto (1,1,1,0).
\end{align*}

\begin{figure}
    \begin{center}
    \includegraphics{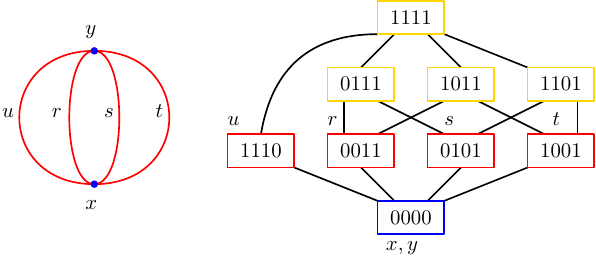}
    \end{center}
    \caption{\emph{Left.} A graph with labeled vertices and edges.
    \emph{Right.} An $\Rbb^4$-filtration of the graph on the left depicted as the Hasse diagram of the subposet of $\Rbb^4$ spanned by grades at which a vertex or an edge appears.}
    \label{fig:graph-field-characteristic}
\end{figure}

\begin{lemma}
    The second and third Betti tables of $H_0(G,f; \kbb)$ depend on the field $\kbb$.
\end{lemma}
\begin{proof}
    By inspection, the graph $(G,f)$ is a minimal graph.
    A straightforward computation shows that the kernel of the morphism $\partial_1^{(G,f)}$ of \cref{definition:H0-H1} is isomorphic to $N_\kbb$; indeed, taking coefficients over the integers, the kernel of $\partial_1^{(G,f)}$ is isomorphic to $N$.
    Thus, by \cref{lemma:R4-example-field-dependence}, the indecomposable summands in the second and third degree of the minimal resolution of $H_0(G,f; \kbb)$ depend on the field $\kbb$, as required.
\end{proof}

\subsection{Field-dependence of higher Betti tables}
Let $\Pscr$ be a finite poset, and let $p \in \Pscr$.
Let $\Ssf_p : \Pscr \to \vect_\kbb$ be the $\Pscr$-persistence module with $\Ssf_p(p) = \kbb$ and $\Ssf_p(r) = 0$ for every $r \neq p \in \Pscr$, with the only possible structure morphisms.

Now, let $n(p) = \min\left( \{p\}^\uparrow \setminus \{p\} \right)$, and define the $\Pscr$-filtered graph $\Gsf_p \coloneqq (V,E,\partial,f)$ by:
\begin{itemize}
    \item $V = \{a,b\}$;
    \item $E = n(p)$;
    \item $\partial : E \to V\times V$ maps every $e \in E$ to $(a,b)$;
    \item $f : V \to \Pscr$ maps $a$ and $b$ to $p$;
    \item $f : E \to \Pscr$ maps $e \in E = n(p) \subseteq \Pscr$ to $e \in \Pscr$.
\end{itemize}

\begin{lemma}
    \label{lemma:hitting-simples}
    Let $\Pscr$ be a finite poset, and let $p \in \Pscr$.
    We have $H_0(\Gsf_p; \kbb) \cong \Ssf_p \oplus \Psf_p$.
\end{lemma}
\begin{proof}
    Let $S \subseteq H_0(\Gsf_p; \kbb) : \Pscr \to \vect_\kbb$ be the submodule generated by $[a] - [b] \in H_0(\Gsf_p;\kbb)(p)$, and let $P \subseteq H_0(\Gsf_p; \kbb) : \Pscr \to \vect_\kbb$ be the submodule generated by $[a] \in H_0(\Gsf_p;\kbb)(p)$.
    It is clear that $S \oplus P = H_0(\Gsf_p; \kbb)$, and that $S \cong \Ssf_p$ and $P \cong \Psf_p$.
\end{proof}

\cref{lemma:hitting-simples} is also related to \cite[Theorem~5.9]{brodzki-burfitt-pirashvili}, which exhibits a large class of persistence modules that can be obtained as a direct summand of the zero dimensional persistent homology of a filtered graph.

Recall that, given persistent modules $M$ and $N$, a projective resolution $C_\bullet \to M$, and $i \in \Nbb$, the vector space $\Ext^i(M,N)$ is defined as the $i$th homology of the chain complex $\Hom(C_\bullet, N)$.
It is a fact that, with this definition, the isomorphism type of $\Ext^i(M,N)$ is independent of the chosen projective resolution.
For more details about this construction, we refer the reader to standard textbooks in homological algebra, such as \cite{rotman}.

For our purposes, the following well-known characterization of Betti tables is sufficient;
for a proof, see, e.g., \cite[Lemma~5.20]{botnan-oppermann-oudot-scoccola}.

\begin{lemma}
    \label{lemma:betti-is-ext}
    Let $\Pscr$ be a finite poset, let $M : \Pscr \to \vect_\kbb$, and let $i \in \Nbb$.
    Then, $\beta_i^M(p) = \dim_\kbb \Ext^i_{\vect_\kbb^\Pscr}(M, \Ssf_p)$.
    \qed
\end{lemma}

The following result is due to Igusa and Zacharia \cite{igusa-zacharia}; it is a variation of a well-known result of Gerstenhaber and Shack, which relates the Hochschild cohomology of a poset algebra to the simplicial cohomology of the geometric realization of the poset \cite{gerstenhaber-shack}.
In order to state the result, we recall a few basic definitions from algebraic topology; for more details, we refer the reader to standard textbooks, such as \cite{dieck}.
The geometric realization of a poset $\Pscr$, denoted $|\Pscr|$, is the geometric realization of the simplicial complex whose simplices are the chains in $\Pscr$ (see \cite[Example~8.1.2]{dieck}).
The $i$th reduced cohomology of a topological space $X$ with coefficients in the field $\kbb$, denoted $\overline{H}^i(X;\kbb)$, is a $\kbb$-vector space, which has the property that, if $X$ and $Y$ are homeomorphic, then $\overline{H}^i(X;\kbb) \cong \overline{H}^i(Y;\kbb)$.

\begin{theorem}[{\cite[Theorem~1.2]{igusa-zacharia}}]
    \label{theorem:igusa-zacharia}
    Let $\Pscr$ be a finite poset, and let $x,y \in \Pscr$ such that $I(x,y) \coloneqq \{z \in \Pscr : y < z < x\}$ is non-empty.
    For every $i \geq 2 \in\Nbb$, we have
    \[
        \Ext^i_{\vect_\kbb^\Pscr}\big(\,\Ssf_{\mathbf x}\,,\,\Ssf_{\mathbf y}\,\big) \; \cong \; \overline{H}^{i-2}\big(\,|I(x,y)|\,;\, \kbb \,\big).
    \]
\end{theorem}

\fielddependence*
\begin{proof}
%
    Let $j \geq 3$.
    Let $\Fbb_2$ be the field of two elements, and recall that, for $n \geq 2$, the real projective space $\Rbb P^n$ is a topological space that has the property that $\overline{H}^{n-1}(\Rbb P^n; \Fbb_2) \cong \Fbb_2$, and $\overline{H}^{n-1}(\Rbb P^n; \kbb) = 0$ if the characteristic of $\kbb$ is not $2$ (see, e.g., \cite[Example~12.2.3]{dieck}).
    Let $\Qscr$ be a finite poset such that $|\Qscr|$ is homeomorphic to the real projective space $\Rbb P^{j-1}$.
    To see that such a poset $\Qscr$ exists, first use that, since $\Rbb P^{j-1}$ is a compact manifold, there exists a finite simplicial complex $K$ whose geometric realization is homeomorphic to $\Rbb P^{j-1}$ \cite[Theorem~7]{whitehead}; then, define $\Qscr$ to be the face poset of $K$, and use the standard fact that $|K|$ is homeomorphic to the geometric realization of $\Qscr$ (see, e.g., \cite[Section~9.3]{bjorner}).
    Now, let $\Pscr$ be given by adjoining to $\Pscr$ a minimum element $q$ and a maximum element $p$.
    In particular, the dimension of $\overline{H}^{j-2}(|I(p,q)| ; \kbb)$ depends on the field $\kbb$.
    We now compute as follows
    \begin{align*}
        \beta_j^{H_0(\Gsf_p; \kbb)}(q) &= \beta_j^{\Ssf_p \oplus \Psf_p}(q)
            = \beta_j^{\Ssf_p}(q) + \beta_j^{\Psf_p}(q)
            = \beta_j^{\Ssf_p}(q)\\
            &= \dim_\kbb \Ext^j_{\vect_\kbb^\Pscr}(\Ssf_p, \Ssf_q)\\
            &= \dim_\kbb \overline{H}^{j-2}(|I(p,q)|; \kbb )\\
            &= \dim_\kbb \overline{H}^{j-2}(\Rbb P^{j-1}; \kbb ),
    \end{align*}
    where in the first line we use the first statement and the easy-to-check properties that Betti tables are additive with respect to direct sums, and that higher Betti tables of a projective module are zero.
    For the fourth and fifth equality we use Lemma \ref{lemma:betti-is-ext} and Theorem \ref{theorem:igusa-zacharia}, and for the last one we use the fact that $|\Qscr|$ is homeomorphic to $\Rbb P^{j-1}$.
    It follows that the $j$th Betti table of $H_0(\Gsf_p; \kbb)$ depends on the field $\kbb$, as required.
\end{proof}

\section{Experiments}
\label{section:experiments}
We have implemented our main algorithms in a C++ package called \emph{graph-mph}~\cite{graph-mph}. We also provide Python bindings for user-friendly visualization of Betti numbers and machine learning applications. The package currently accepts either a point cloud or a bifiltration, and outputs a minimal presentation and absolute Betti numbers. 

In the computational setting, we work with a finite (potentially multi-critical) bifiltration $F$.
Such a bifiltration can be specified by a single simplicial complex $K$ together with a collection of incomparable grades $\mathrm{births}(\sigma)\subset\mathbb R^2$ for each simplex $\sigma$, specifying the bigrades where it appears.
We do not requite the input filtration to be 1-critical ($\mathrm{births}(\sigma)$ does not have to contain just one element) instead we implement the algorithm in \cref{section:multi-critical} that reduces multi-critical filtrations to satisfy this requirement.

The bifiltrations in our experiments are the degree-Rips and function-Rips bifiltrations, which we now recall. 

\paragraph{Function-Rips bifiltration} For a finite metric space $P$ with metric $d$ and $r\geq 0$, let $N(P)_r$ denote the $r$-neighborhood graph of $P$: the vertex set of $N(P)_r$ is $P$, and there is an edge $\{i,j\}\in N(P)_r$ if and only if $d(i,j)\leq r$. If $r<0$, we define $N(P):=\emptyset$. Given any function $\gamma:P\to \mathbb R$, we define the function-Rips bifiltration by $FR(\gamma)_{a,b}:=N(\gamma^{-1}(-\infty,a])_b$. Note that $FR(\gamma)$ is always 1-critical. 

Apart from taking a user-defined function $\gamma$, our package also supports constructing a built-in ball density function directly from a point cloud. It is defined by counting the number of points in $P$ within distance $r$ of $x$, that is
 $\gamma(x)=C\cdot | \{ p \in P \mid d(x,p) \leq r \} |$,
where $r>0$ is a fixed parameter (by default we take the 20-th percentile of all pairwise distances in the point cloud), and $C$ is a normalization constant chosen so that $\sum_{x\in P} \gamma(x)=1$.

\paragraph{Degree-Rips bifiltration} For $r,d\in \mathbb R$, let $P_{d,r}\subset P$ be the set of vertices in $N(P)_r$ of degree at least $d$. We define the degree-Rips bifiltration $DR(P)$ by taking $DR(P)_{d,r}:= N(P_{d,r})_r$, the $r$-neighborhood graph of $P_{d,r}$. Note that this is in fact a bifiltration indexed by $\mathbb R^{\mathrm{op}}\times \mathbb R$, where $\mathbb R^{\mathrm{op}}$ denotes the opposite poset of $\mathbb R$; that is, $DR(P)_{a,b}\subset DR(P)_{a',b'}$ whenever $a\geq a'$ and $b\leq b'$. If $P$ has more than one point, then $DR(P)$ is multi-critical. In order to use our algorithm, we convert it into a 1-critical bifiltration as described in \cref{section:multi-critical} and negate its first coordinate so that it is indexed by $\mathbb R^2$.

\subsection{Comparison with RIVET and mpfree}
We compare our \emph{graph-mph} with two other packages: RIVET \cite{rivet} and mpfree \cite{kerber-rolle}. Note that RIVET and graph-mph both support point cloud and bifiltration input; they output the minimal presentation and Betti numbers. However, mpfree only accepts bifiltration input (which forces the file size to be considerably larger than a point cloud) and outputs the minimal presentation. In order to make a fair comparison, we have removed the I/O time in the following discussion. The datasets we use are:

\begin{itemize}
    \item 2D Annulus with noise. We generate it by adding Gaussian noise to the points around a circle. 
    \item 3D Blob clusters. This dataset contains three clusters and is generated by scikit-learn \cite{scikit-learn}. As shown in \cref{fig:3D-blob-clusters}, the purple and yellow clusters are closer to each other than the green one.
    \item NYC Taxi pickups dataset from January 2025. This dataset is obtained from the official website and is publicly available\footnote{\url{https://www.nyc.gov/site/tlc/about/tlc-trip-record-data.page}}. The complete dataset is a table with 3.4 million rows and 20 columns, where each row represents one taxi ride including pickup and dropoff times and locations, number of passengers, tip and fare amounts, etc. For our experiment, we use only the pickup locations and merge them by zones. This results in a point cloud with a function by assigning to the centroid of each zone the number of pickups in that zone, as shown in \cref{fig:NYC-taxi-pickups}.
\end{itemize}

\begin{figure}
    \centering
    \includegraphics[width=0.5\linewidth]{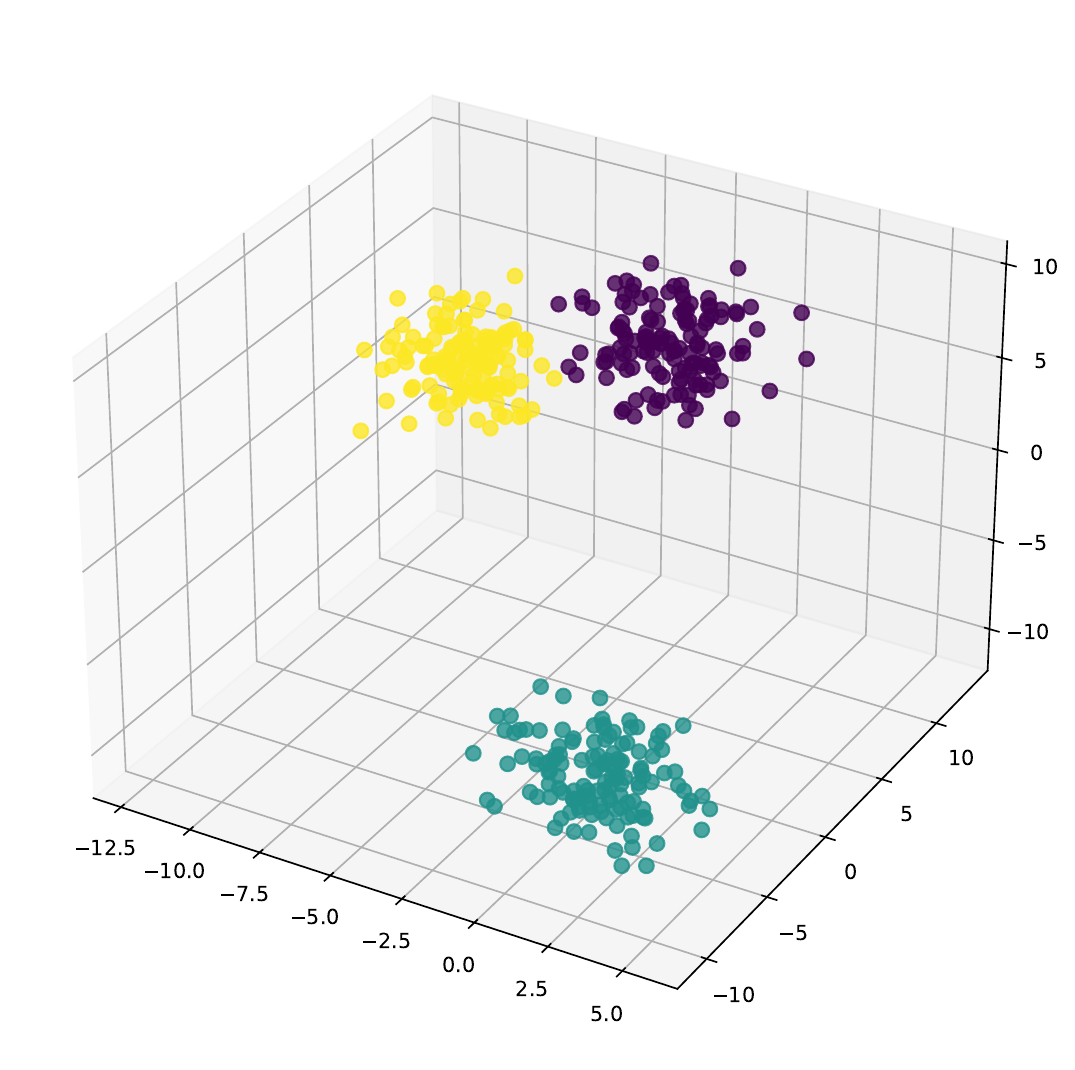}
    \caption{%
    The 3D Blob clusters of size 400. The purple and yellow clusters are closer to each other than the green one.
    }
    \label{fig:3D-blob-clusters}
\end{figure}

\begin{table}[ht]
    \centering
    \renewcommand{\arraystretch}{1.2}
    \begin{tabular}{c|cccc|cccc}
    \hline
    & \multicolumn{4}{c|}{ball-density-Rips} & \multicolumn{4}{c}{deg-Rips} \\ \hline
    size & our-float & our-rational & RIVET & mpfree & our-float & our-rational & RIVET & mpfree \\ \hline
    200  & 0.038 & 0.05 & 0.079 & \textbf{0.026} & \textbf{0.48} & 0.59 & 2.94 & 1.97 \\ \hline
    400  & \textbf{0.10} & 0.17 & 0.46 & 0.12 & \textbf{4.59} & 5.26 & 38.56 & 17.07 \\ \hline
    600  & \textbf{0.24} & 0.45 & 1.59 & 0.33 & \textbf{17.36} & 20.44 & 184.30 & 60.97 \\ \hline
    800  & \textbf{0.46} & 0.86 & 4.33 & 0.86 & \textbf{51.53} & 54.34 & 625.99 & -- \\ \hline
    1000 & \textbf{0.77} & 1.49 & 8.20 & 1.51 & \textbf{122.03} & 124.12 & -- & -- \\ \hline
    \end{tabular}
    \vspace{0.3cm}
    \caption{Time comparison on 2D noisy annulus of different sizes. Our algorithm is implemented in both float and rational data types (same as RIVET). The dash `--' means that the program terminated due to memory overflow.
    }
    \label{table:time-2D-noisy-annulus}
\end{table}

In \cref{table:time-2D-noisy-annulus}, we compare the time taken by our algorithm on 2D noisy annulus of different sizes. We use the ball-density-Rips and degree-Rips bifiltrations. In order to match the precision of RIVET, we implement our algorithm in the rational data type in addition to the standard float data type. Because converting the point cloud from floating point to rational data type creates additional overhead, float results are faster than rational. The results show that our algorithm is faster than RIVET and mpfree. It is worth noting that our algorithm is memory-efficient and can handle large datasets.
The dash `--' indicates that the program terminated due to memory overflow, which happens for both RIVET and mpfree.

\begin{table}[ht]
    \centering
    \renewcommand{\arraystretch}{1.2}
    \begin{tabular}{c|c|c|c|c|c}
    \hline
    Size & Grade Table size (in millions) & Active Grade Size & our & rivet & mpfree \\ \hline
    500 & 62m  & 0.49m (0.79\%) & \textbf{9.41} & 90.21 & 35.09 \\ \hline
    600 & 108m & 0.71m (0.66\%) & \textbf{17.05} & 185.33 & 61.57 \\ \hline
    700 & 171m & 0.97m (0.56\%) & \textbf{29.21} & 348.43 & 101.75 \\ \hline
    800 & 256m & 1.27m (0.49\%) & \textbf{46.54} & 626.19 & -- \\ \hline
    900 & 364m & 1.61m (0.44\%) & \textbf{70.18} & 1159.92 & -- \\ \hline
    1000 & 499m & 1.99m (0.39\%) & \textbf{102.67} & -- & -- \\ \hline
    \end{tabular}
    \vspace{0.3cm}
    \caption{Time spent on computing minimal-presentation of degree-Rips bifiltrations on 3D Blob of 3 clusters of different sizes.}
    \label{table:time-3D-blob-clusters}
\end{table}

We also see in \cref{table:time-3D-blob-clusters} that our algorithm is faster than RIVET and mpfree on 3D Blob clusters of different sizes. The two additional columns are the grade table size and the active grade size. Grade table size represents the product of the number of $x$-coordinates and $y$-coordinates in the grade table. Active grade size represents the percentage of the grade table that is active, i.e., the grade points that are held by at least one vertex or edge. It shows the sparsity of the grade table as the size of the dataset increases.
    
\subsection{Visualization}

Based on Betti numbers, we can also visualize the Hilbert function, $\Hil^M:\mathbb R^2\to \mathbb N$ which sends $a$ to $\dim M_a$, the number of connected components of the bifiltered graph at $a$. It can be computed from the Betti numbers~\cite{peeva} in linear time via 2D prefix sum:
$$
\Hil^M(a) = \sum_{b\leq a} \beta_0(b) - \beta_1(b) + \beta_2(b)
$$

For example, for the bifiltered graph in \cref{fig:graded-graph}, we show the three persistence modules in a minimal resolution of $H_0(G,f)$ in \cref{fig:graded-graph-betti-numbers} and its Hilbert function in \cref{fig:Hilbert_matrix_with_betti_numbers_figure1}.

\begin{figure}
    \centering
    \includegraphics[width=0.32\linewidth]{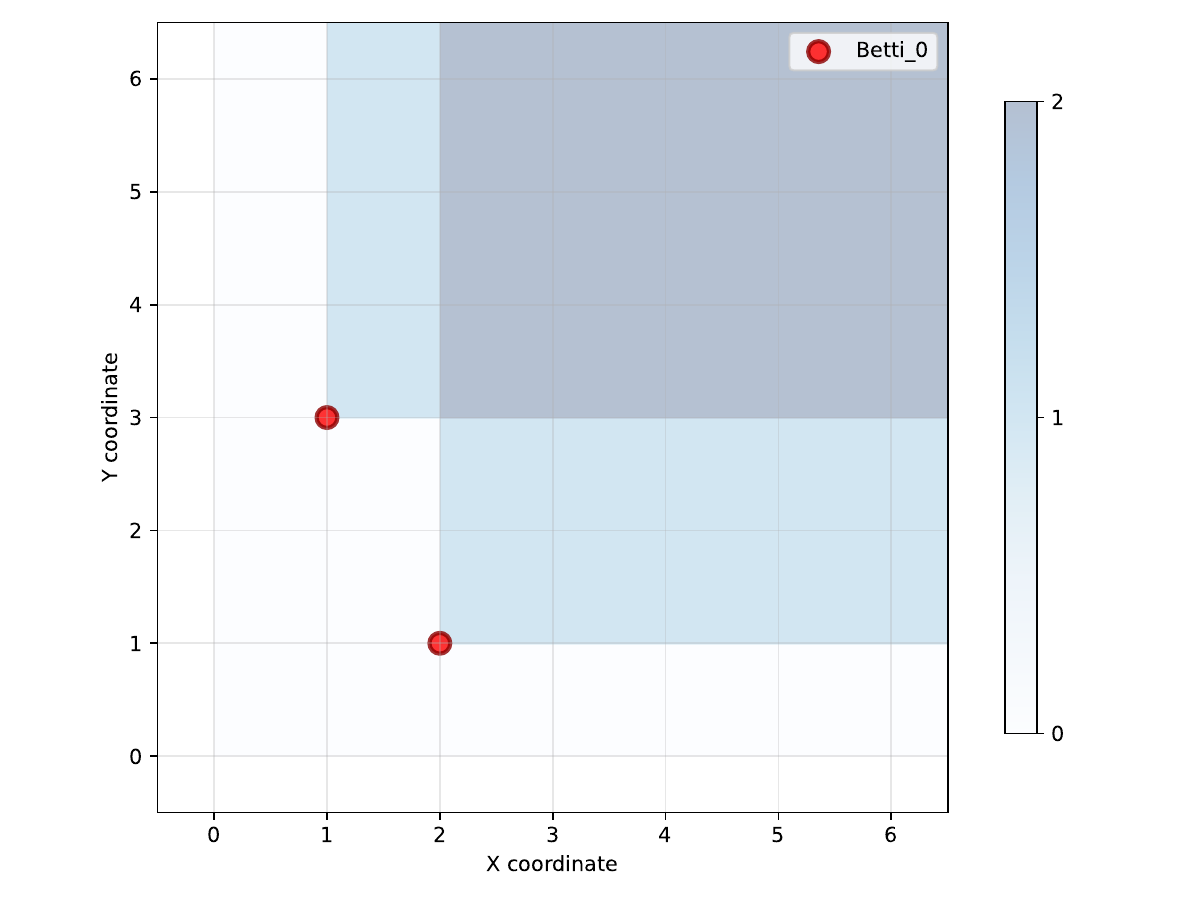}\hfill
    \includegraphics[width=0.32\linewidth]{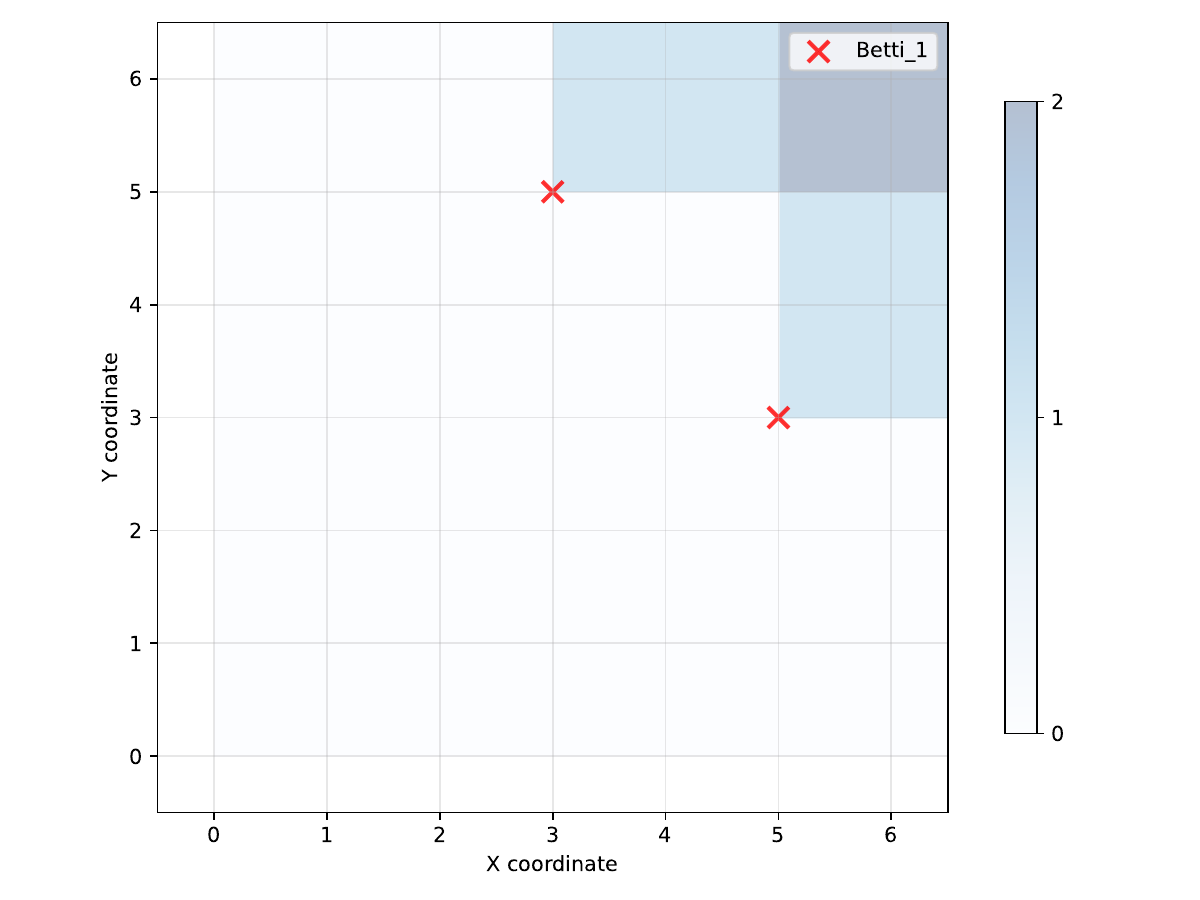}\hfill
    \includegraphics[width=0.32\linewidth]{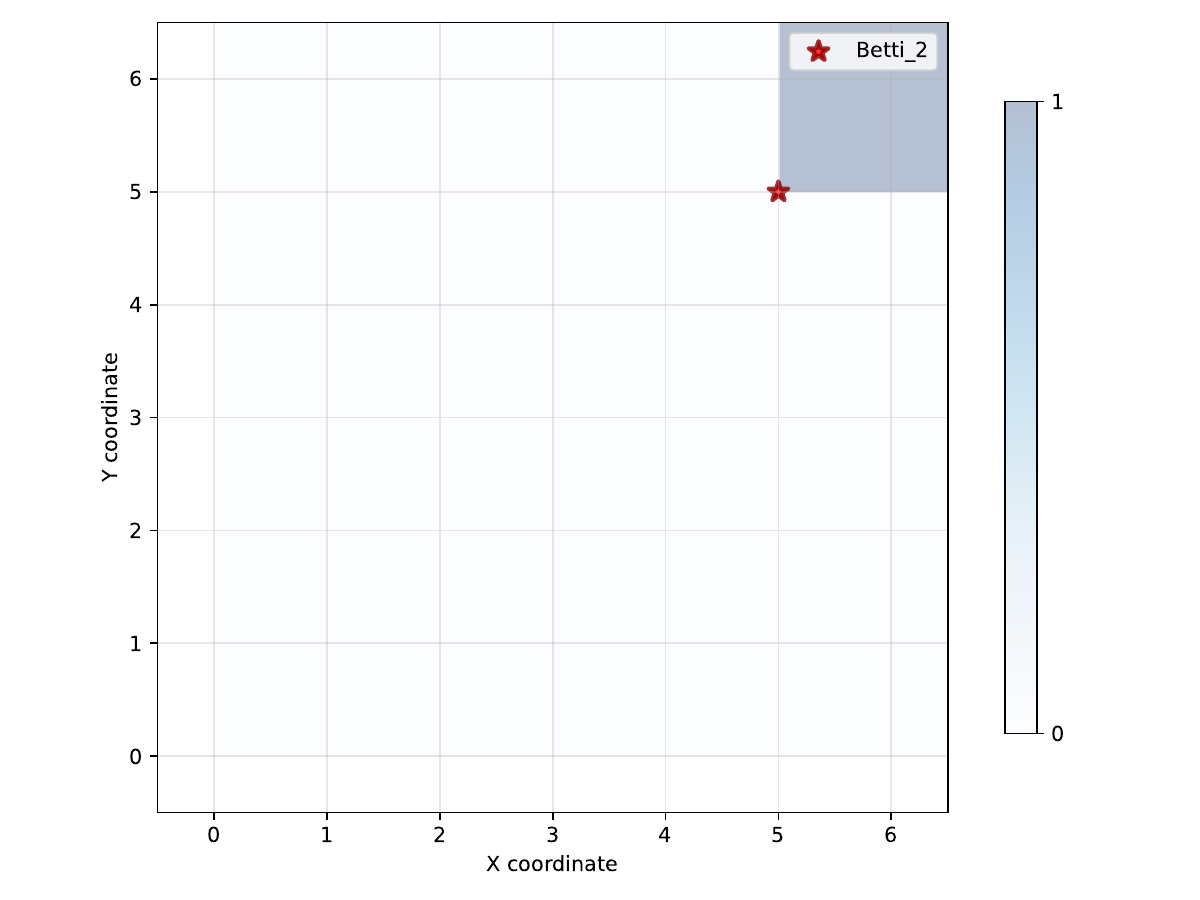}
    \caption{%
    The three persistence modules in a minimal resolution of $H_0(G,f)$ for the bifiltered graph in \cref{fig:graded-graph}.
    }
    \label{fig:graded-graph-betti-numbers}
\end{figure}

\begin{figure}
    \centering
    \includegraphics[width=0.5\linewidth]{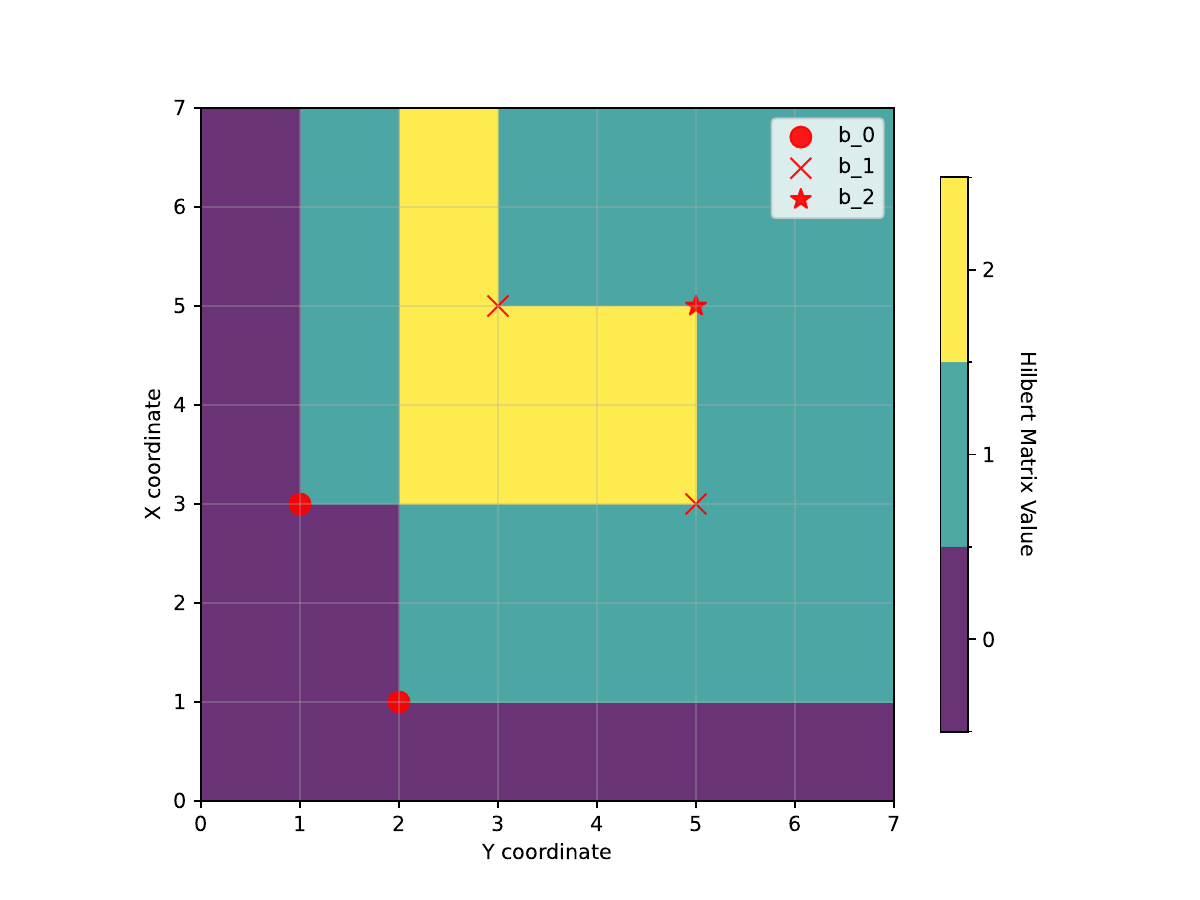}
    \caption{%
    Hilbert matrix of the bifiltered graph in \cref{fig:graded-graph}.
    }
    \label{fig:Hilbert_matrix_with_betti_numbers_figure1}
\end{figure}

In \cref{fig:Hilbert_matrix_blob_400}, we show the two Hilbert matrices with Betti numbers using ball-density-Rips and degree-Rips bifiltrations on 3D Blob clusters of size 400. For visualization purposes, we clip values larger than 5 to 5. Both matrices have a large region with 2 connected components, and the region with 3 connected components is comparatively smaller, which indicates that the point cloud has two close clusters and one far away. This also matches \cref{fig:3D-blob-clusters}, where the purple and yellow clusters are closer to each other than the green one.

In \cref{fig:NYC-taxi-pickups}, we show the Hilbert matrix from Betti numbers computed from functional-Rips bifiltration on the NYC Taxi dataset. The top figure shows the NYC Taxi pickups by zone, the middle figure is the converted point cloud, and the bottom figure shows the Hilbert matrix with Betti numbers. It shows that the Hilbert matrix has a large region of value 2 and a smaller region of value 3 merging into the region of value 2 when the distance increases to 30,000. The large region of value 2 might correspond to the fact that John F. Kennedy airport has a large number of pickups and is far away from other zones. The small region of value 3 might correspond to the fact that LaGuardia airport has a high number of pickups and is close to other zones. 

\begin{figure}
    \centering
    \includegraphics[width=0.5\linewidth]{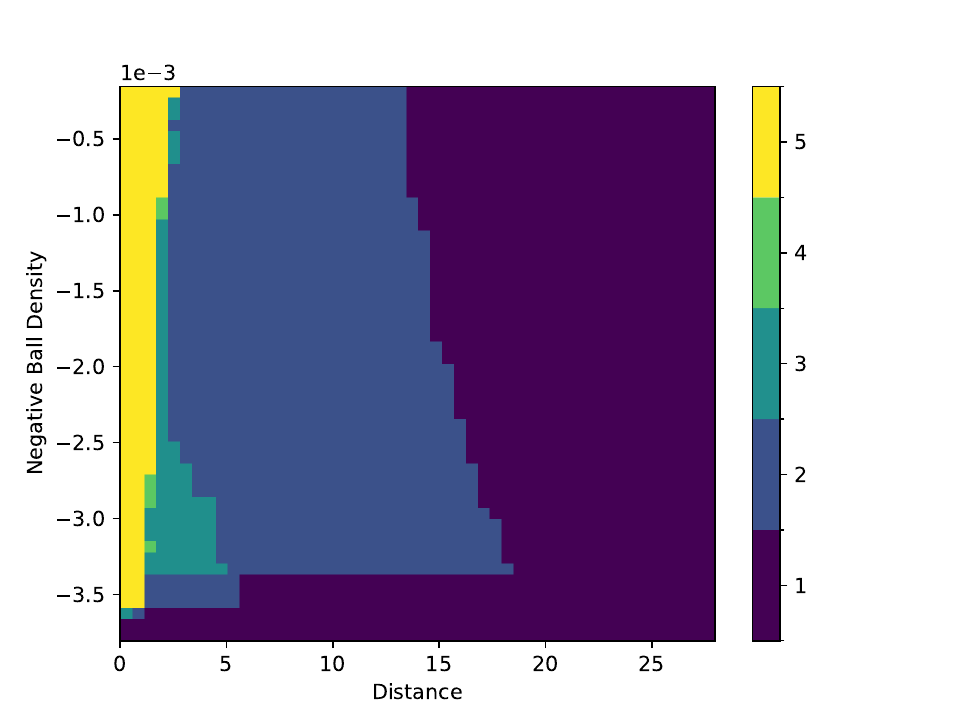}\hfill
    \includegraphics[width=0.5\linewidth]{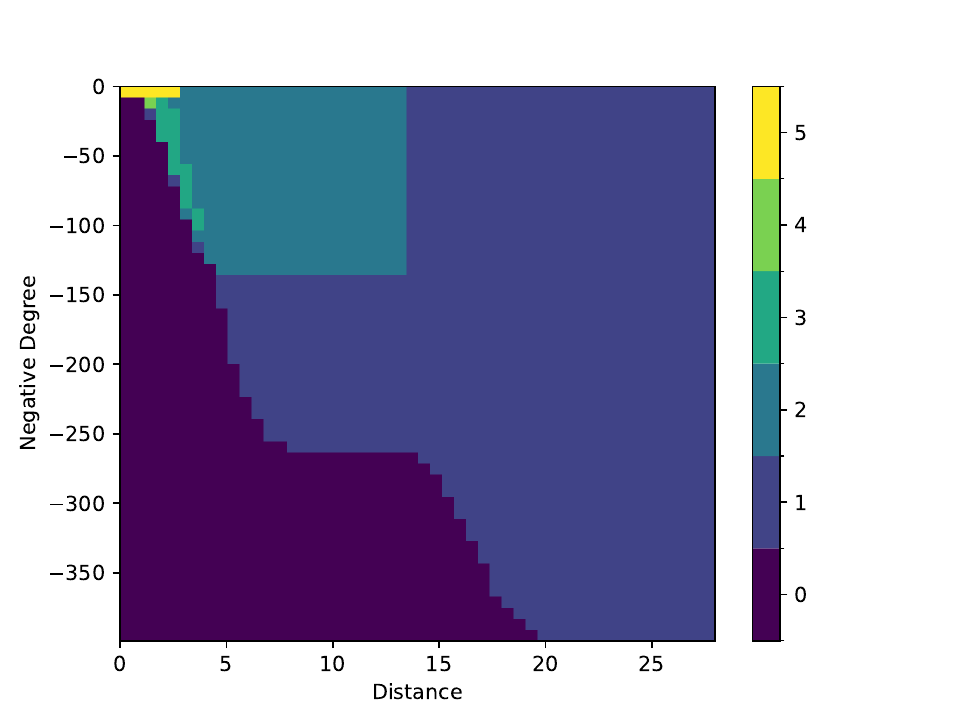}
    \caption{
        Hilbert matrices with Betti numbers for ball-density-Rips (Left) and degree-Rips (Right) bifiltrations on 3D Blob clusters of size 400 in \cref{fig:3D-blob-clusters}. For visualization purposes, values larger than 5 are clipped to 5. Both Hilbert matrices have a large region of value 2, and a smaller region of value 3 indicates that the point cloud has two close clusters. 
    }
    \label{fig:Hilbert_matrix_blob_400}
\end{figure}

\begin{figure}
    \centering
    \includegraphics[width=0.6\linewidth]{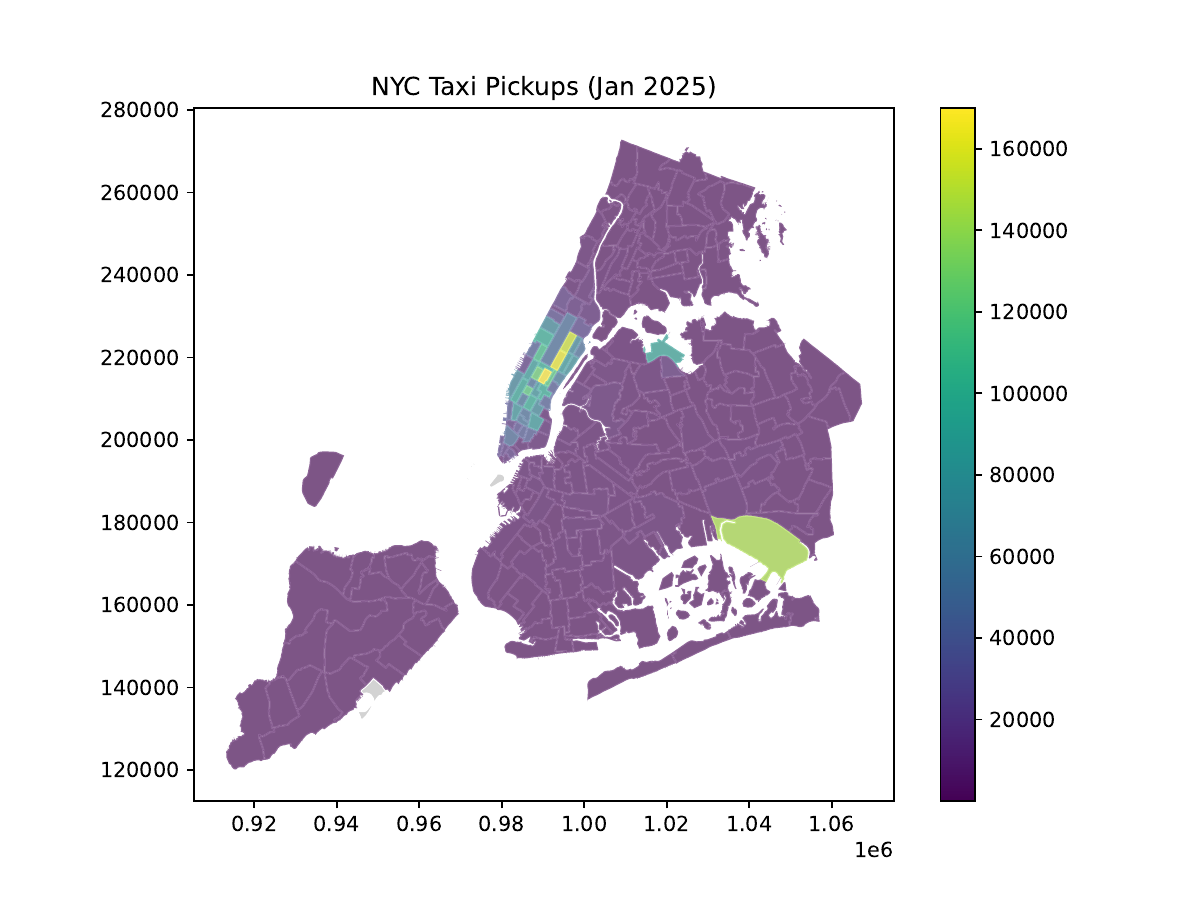}\hfill
    \includegraphics[width=0.6\linewidth]{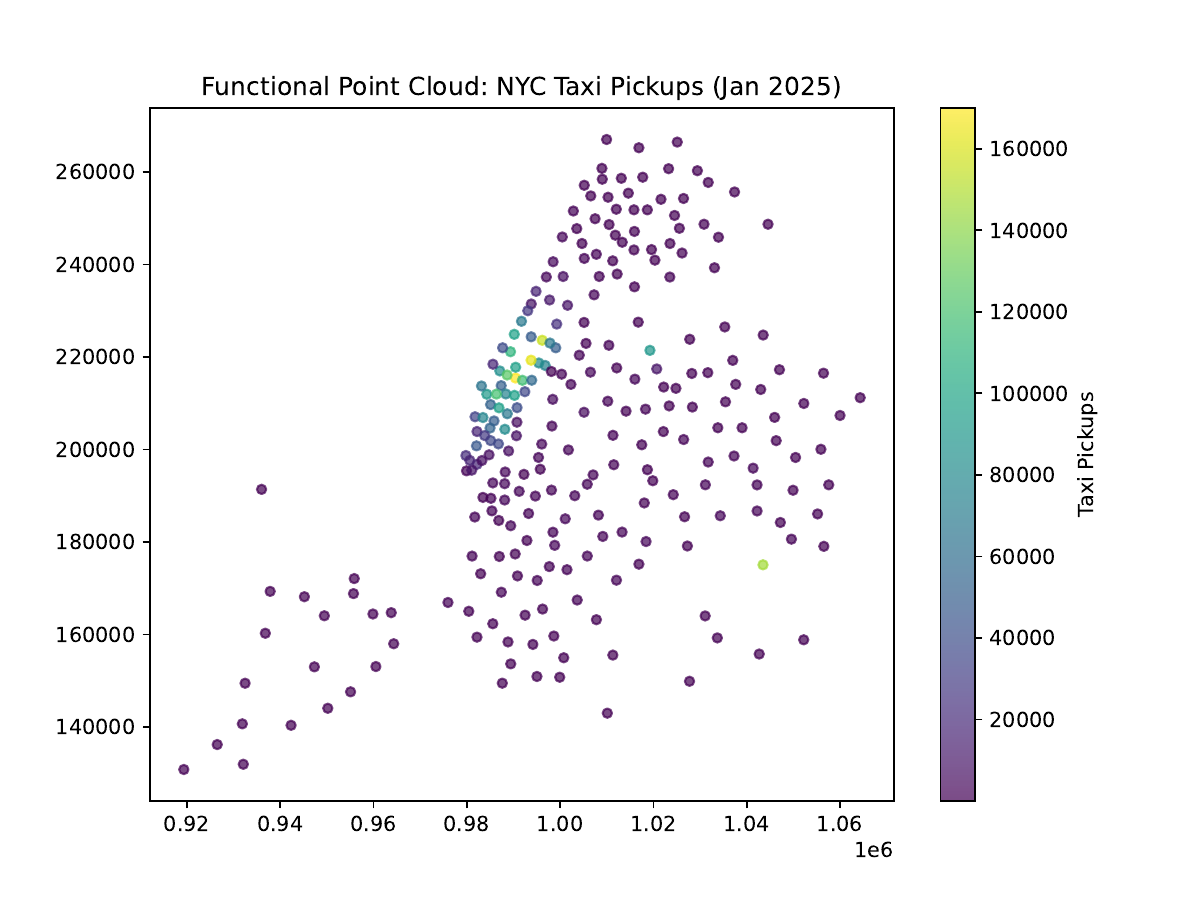}\hfill
    \includegraphics[width=0.6\linewidth]{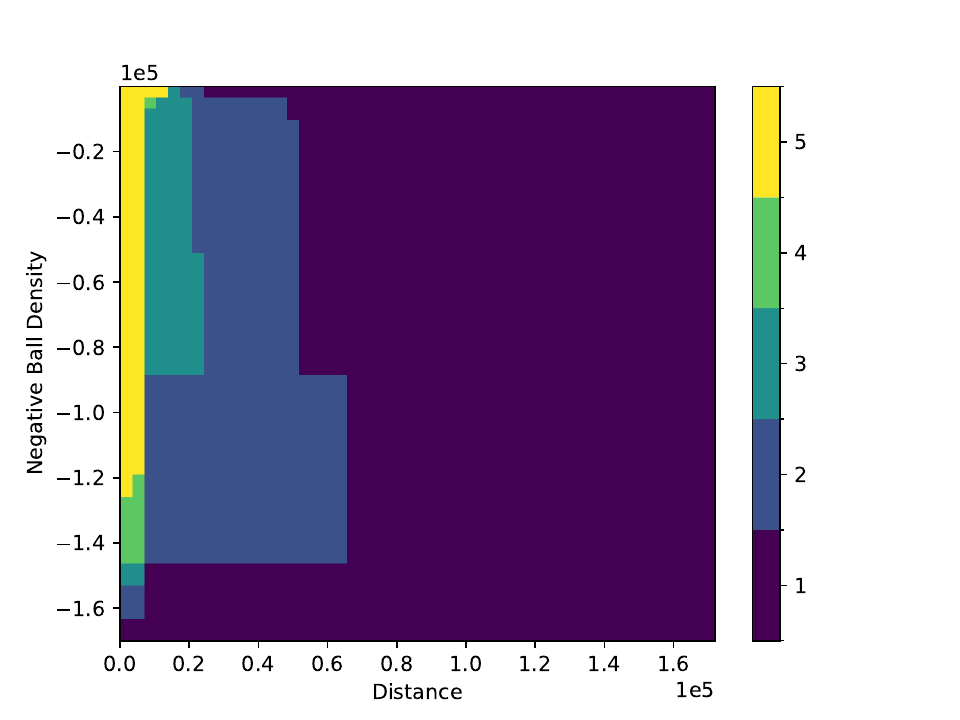}
    \caption{Hilbert matrix with Betti numbers for ball-density-Rips bifiltration on NYC Taxi dataset. \emph{Top:} NYC Taxi pickups by zone. \emph{Middle:} NYC Taxi pickups as a point cloud. \emph{Bottom:} Hilbert matrix with Betti numbers.}
    \label{fig:NYC-taxi-pickups}
\end{figure}

\subsection{Detailed timing analysis}
Using the same time recording framework as \cite{kerber-rolle}, we record and provide a detailed time analysis of our algorithm and mpfree. For example, to compute a minimal presentation matrix of the degree-Rips bifiltration of the 3D Blob clusters of size $700$ (detailed time output is in Appendix \ref{section:raw-time-output}), out of the total time $101$ seconds, mpfree takes $81$ seconds to convert the raw input into its (dense) graded matrices. Instead, our algorithm takes $14$ seconds ($49\%$ of the total time) to sort all grades from vertices and edges lexicographically before the main loop (line 5) in \cref{algorithm:betti-tables}. We achieve this by iterating over all vertices and edges after collapsing, and storing and sorting them lexicographically in a single container. Although both algorithms spend a majority of time on processing or sorting the grade points, our algorithm delays the sorting after collapsing, which is both time and memory efficient.

Another optimization we apply is to store the graph as adjacency lists with duplicates. At the end of \Cref{algorithm:vertex-reduction,algorithm:local-components}, one important overhead is the update of the graph after collapsing. Keeping the duplicate does not affect the correctness of the output because the DFS keeps track of the visited vertices.

%
%

\section*{Acknowledgments}
We thank the anonymous reviewers and the SoCG'25 program committee for helpful comments and for spotting an issue with Definition \ref{definition:cycle-creating}, and its corresponding fix.

YL and DM were supported by the National Science Foundation, award DMS-2324632.
LS was supported by the National Science Foundation through grants CCF-2006661 and CAREER award DMS-1943758, while at Northeastern University, and by EPSRC grant ``New Approaches to Data Science: Application Driven Topological Data Analysis'', EP/R018472/1, while at University of Oxford.

\appendix

\section{Technical lemmas}
\label{section:proofs}

\begin{lemma}
    \label{lemma:projective-cover-minimal-elements}
    Let $M$ be a $\Pscr$-persistence module and let $g : Q = \bigoplus_{i \in I} \Psf_{f(i)} \cdot \{i\} \to M$ be a finite projective cover.
    For every $j \in I$, we have that $g_{f(j)}(\{j\}) \neq 0 \in M(f(j))$, and $\{j\} \in Q(f(j))$ is the only preimage of $g_{f(j)}(\{j\})$.
\end{lemma}
\begin{proof}
    If $g_{f(j)}(\{j\}) = 0 \in M(f(j))$, then the restriction of $g$ given by
    $\bigoplus_{i \in I \setminus \{j\}} \Psf_{f(i)} \cdot \{i\} \to M$ would be surjective, and thus $\sum_{r \in \Pscr} \beta^Q(r)$ would not be minimal, so $Q \to M$ would not be a projective cover.

    Similarly, if there exists another element of $Q(f(j))$ mapping to $g_{f(j)}(\{j\})$, then this element would be linearly independent of $\{j\}$, and again the restriction of $g$ given by $\bigoplus_{i \in I \setminus \{j\}} \Psf_{f(i)} \cdot \{i\} \to M$ would be surjective, and $Q \to M$ would not be a projective cover.
\end{proof}

\begin{lemma}
    \label{lemma:deletable-pi-0-iso}
    If $(G,f) = (V,E,\partial,f)$ is a filtered graph and $e \in E$ is deletable, then the morphism of persistent sets $\pi_0(G \setminus e,f) \to \pi_0(G,f)$ induced by the inclusion $G \setminus e \to G$ is an isomorphism.
\end{lemma}
\begin{proof}
    The morphism $\pi_0(G \setminus e,f) \to \pi_0(G,f)$ is clearly surjective, since $G \setminus e \to G$ is surjective on vertices.
    Thus, it is sufficient to prove that, for all $r \in \Pscr$ and all $x,y \in V$ with $f(x),f(y) \geq r$, we have that $[x] = [y] \in \pi_0(G,f)(r)$ implies $[x] = [y] \in \pi_0(G \setminus e,f)(r)$.
    So assume that $[x] = [y] \in \pi_0(G,f)(r)$, and let $e^1, \dots, e^k$ be a path between $x$ and $y$ in $(G,f)_r$.
    If the path does not contain the edge $e$, then it is also a path between $x$ and $y$ in $(G\setminus e, f)_r$, and we are done.
    If the path does contain the edge $e$, then, in particular, $r \geq f(e)$.
    Since $e$ is deletable, there exists a path $d^1, \dots, d^\ell$ between the vertices $e_0$ and $e_1$ in $(G \setminus e, f)_{f(r)}$.
    Now, take the path $e^1, \dots, e^k$ and replace every occurrence of the edge $e$ with the path $d^1, \dots, d^\ell$.
    This yields a path between $x$ and $y$ in $(G\setminus, f)_r$, so that $[x]=[y] \in \pi_0(G \setminus e,f)_r$, as required.
\end{proof}

\section{Detailed Time Output}\label{section:raw-time-output}

The following is the raw time output of mpfree and our algorithm for computing a minimal presentation of the degree-Rips bifiltration of the 3D Blob clusters of size 700. The time is in seconds. Both algorithms use the same timer tool provided by the boost library.

For mpfree, we added two timers beyond those it already provides: ``Load data into prematrix'' and ``Prematrix to graded matrix,'' both reported within the original ``IO timer.'' Unlike our algorithm, which takes a point cloud as input, mpfree requires the input to be a bifiltration, resulting in considerably larger file sizes. Its pipeline first loads the bifiltration into prematrices and then converts them into graded matrices. Since the conversion step does not involve I/O, we report it separately.

For our algorithm, the timers follow the order of steps in \cref{algorithm:betti-tables}. ``Build filtration graph'' measures the time to construct the graph from the bifiltration, with ``Load input'' as a sub-timer for reading and parsing the input file. ``Collapse edge'' and its sub-timer ``Update graph'' correspond to \cref{algorithm:local-components}. ``Collapse vertex'' and its sub-timer ``Update graph'' correspond to \cref{algorithm:vertex-reduction}, excluding the edge-collapsing step (line~1). ``Collect grades lexicographically'' measures the time to collect and sort all grades before the main loop (line~5 of \cref{algorithm:betti-tables}). ``Main loop (visit grades)'' measures the time to iterate through the grades and compute the Betti tables. As with mpfree, we also report the total time excluding input loading for a fairer comparison.

Note that both mpfree (`Prematrix to graded matrix timer') and our algorithm (`Collect grades lexicographically') spend a significant amount of time on sorting the grades lexicographically. 


\begin{table}[h]
\centering
    \begin{tabular*}{0.85\linewidth}{@{\extracolsep{\fill}}lrr}
    \hline
    \textbf{mpfree} & \textbf{Time (s)} & \textbf{Percentage} \\
    \hline
    IO timer & 116.701 & 86.30\% \\
    \quad Load data into prematrix & 33.480 & 24.76\% \\
    \quad Prematrix to graded matrix & 81.582 & 60.33\% \\
    Chunk timer & 12.255 & 9.06\% \\
    Mingens timer & 3.034 & 2.24\% \\
    Kerbasis timer & 0.004 & $<$0.01\% \\
    Reparam timer & 0.003 & $<$0.01\% \\
    Minimize timer & 0.007 & $<$0.01\% \\
    \hline
    \textbf{Total time} & \textbf{135.233} & --- \\
    Total time (excl.\ load data into prematrix) & 101.753 & --- \\
    \hline
    \end{tabular*}
\end{table}
    
\begin{table}[h]
\centering
    \begin{tabular*}{0.85\linewidth}{@{\extracolsep{\fill}}lrr}
    \hline
    \textbf{Our algorithm} & \textbf{Time (s)} & \textbf{Percentage} \\
    \hline
    Build filtration graph & 5.628 & 19.25\% \\
    \quad Load input & 0.010 & 0.03\% \\
    Collapse edge & 1.073 & 3.67\% \\
    \quad Update graph & 0.187 & 0.64\% \\
    Collapse vertex & 5.198 & 17.79\% \\
    \quad Update graph & 0.703 & 2.41\% \\
    Collect grades lexicographically & 14.586 & 49.91\% \\
    Main loop (visit grades) & 2.517 & 8.61\% \\
    \hline
    \textbf{Total time} & \textbf{29.227} & --- \\
    Total time (excl.\ load input) & 29.217 & --- \\
    \hline
    \end{tabular*}
\end{table}

\eject

\bibliographystyle{amsplain}
\bibliography{references}

\end{document}